%% file: main.tex
\documentclass[lettersize,journal]{IEEEtran}

\usepackage{cite}
\usepackage{amsmath,amsfonts,amssymb}
\usepackage{algorithmicx}
\usepackage{algpseudocode}
\usepackage{algorithm}
\usepackage{array}
\usepackage[caption=false,font=normalsize,labelfont=sf,textfont=sf]{subfig}
\usepackage{textcomp}
\usepackage{stfloats}
\usepackage{url}
\usepackage{verbatim}
\usepackage{graphicx}
\usepackage{cite} 

\usepackage{times}  
\usepackage{helvet} 
\usepackage{courier}
\usepackage{colortbl}
\usepackage{hhline}
\usepackage[T1]{fontenc}
\usepackage{rotating}
\usepackage{microtype}
\usepackage{xcolor}
\usepackage{multirow}
\usepackage[utf8]{inputenc}
\usepackage{booktabs}     
\usepackage{nicefrac}   
\usepackage{latexsym}
\usepackage{tabularx}
\usepackage{inconsolata}
\usepackage{newfloat}
\usepackage{listings}
\usepackage{bibentry}
\usepackage{amsthm}
\usepackage{thmtools}
\usepackage{orcidlink}
\usepackage{hyperref}

\usepackage{tikz}
\usetikzlibrary{calc}
\usepackage{pgfmath}
\definecolor{monte_carlo}{RGB}{103,194,172}

\theoremstyle{plain}
\newtheorem{theorem}{Theorem}[section]
\newtheorem{lemma}[theorem]{Lemma}
\newtheorem{proposition}[theorem]{Proposition}

\theoremstyle{definition}

\newtheorem{assumption}[theorem]{Assumption}

\theoremstyle{remark}

\begin{document}

\title{Membership Inference for Contrastive Pre-training Models with Text-only PII Queries}

\author{
    Ruoxi Cheng\textsuperscript{*}~\orcidlink{0009-0001-3261-642X}, 
    Yizhong Ding\textsuperscript{*}~\orcidlink{0009-0001-1339-9597}, 
    Jian Zhao\textsuperscript{*}~\orcidlink{0000-0002-3508-756X}, 
    Hongyi Zhang~\orcidlink{0009-0004-1415-4761}, 
    Haoxuan Ma~\orcidlink{0009-0003-9159-8764},  \\
    Tianle Zhang\textsuperscript{\dag}~\orcidlink{0000-0003-4881-2406},
    Yiyan Huang\textsuperscript{\dag}~\orcidlink{0000-0003-1268-2208},
    and Xuelong Li~\orcidlink{0000-0002-0019-4197},~\IEEEmembership{Fellow,~IEEE}
\thanks{\textsuperscript{*}Co-first authors. \textsuperscript{\dag}Corresponding authors.}%
\thanks{R. Cheng is with the Beijing Electronic Science and Technology Institute, Beijing 100070, China, and also with the Institute of Artificial Intelligence, China Telecom (TeleAI), Beijing 100032, China.}%
\thanks{Y. Ding is with the Beijing Electronic Science and Technology Institute, Beijing 100070, China.}%
\thanks{J. Zhao, T. Zhang, and X. Li are with the Institute of Artificial Intelligence, China Telecom (TeleAI), Beijing 100032, China. (e-mail: zhangtl15@chinatelecom.cn)}%
\thanks{H. Zhang is with Nanyang Technological University, Singapore 637371, Singapore.}%
\thanks{H. Ma is with Nanjing University, Nanjing 210023, China.}%
\thanks{Y. Huang is with the School of Computing and Information Technology, Great Bay University, Dongguan 523000, China. (e-mail: huangyiyan@gbu.edu.cn)}%
}

\markboth{Preprint, under review.}%
{Cheng \MakeLowercase{\textit{et al.}}: Membership Inference for Contrastive Pre-training Models with Text-only PII Queries}


\maketitle

\begin{abstract}
Contrastive pretraining models such as CLIP and CLAP, serve as the ubiquitous perceptual backbones for modern multimodal large models, yet their reliance on web-scale data raises growing concerns about memorizing Personally Identifiable Information (PII). Auditing such models via membership inference is challenging in practice: shadow-model MIAs are computationally prohibitive for large multimodal backbones, and existing multimodal auditing methods typically require querying the target with paired biometric inputs, thereby directly exposing sensitive biometric information to the target model. To bypass this critical limitation, we demonstrate a highly desirable capability for privacy auditing: multimodal memorization within these foundational encoders can be accurately inferred using exclusively the text modality. We propose \textbf{U}nimodal \textbf{M}embership \textbf{I}nference \textbf{D}etector (UMID), a text-only auditing framework that performs text-guided cross-modal latent inversion and extracts two complementary signals, \textit{similarity} (alignment to the queried text) and \textit{variability} (consistency across randomized inversions). UMID compares these statistics to a lightweight non-member reference constructed from synthetic gibberish and makes decisions via an ensemble of unsupervised anomaly detectors.
Comprehensive experiments across diverse CLIP and CLAP architectures demonstrate that UMID significantly improves the effectiveness and efficiency over prior MIAs, delivering strong detection performance with sub-second auditing cost using solely text queries, completely circumventing the need for biometric inputs and complying with strict privacy constraints.
\end{abstract}



\maketitle

\input{1_introduction_revised}

\input{2_relatwork_revised}

\input{3_methodology_revised}

\input{4_evaluation_revised}

\input{5_defense_revised}

\input{6_conclusion}

\input{7_contributions}


\bibliographystyle{IEEEtran}
\bibliography{main}


\input{9_appendix}

\end{document}

%% file: 1_introduction_revised.tex
\section{Introduction}
\label{sec:intro}
Contrastive foundation models \cite{yuan2021multimodal}, exemplified by CLIP \cite{radford2021learning} and CLAP \cite{elizalde2023clap}, have become the ubiquitous perceptual backbones for modern Multimodal Large Language Models (MLLMs) \cite{cheng2025ecoalign} and various vision/audio--language systems. Their strong performance, however, is largely enabled by web-scale training data that can contain Personally Identifiable Information (PII) \cite{schwartz2011pii,sun2025better}, raising increasing concerns about PII leakage and downstream misuse \cite{xi2024defending,hu2023defenses,pei2025selfprompt,zhao2025strata}. Recent studies further suggest that multimodal models may inadvertently memorize membership signals or identity attributes, exposing individuals to privacy risks when sensitive biometric data (e.g., faces or voices) appear in the training set \cite{10095737}. Consequently, auditing these multimodal models for membership leakage has become a critical imperative, not only to meet regulatory requirements, but also to provide an accessible safeguard that enables the public to assess potential data exposure in widely deployed foundation models without technical barriers \cite{li2025tuni}.

Current auditing protocols primarily rely on Membership Inference Attacks (MIAs) \cite{shokri2017membership}, which aim to determine whether a specific data sample was used for model training. However, when applied to contrastive pretraining models, existing MIAs face two critical challenges. \textbf{(1) Computationally prohibitive costs:} standard MIAs require training multiple shadow models to approximate the target model’s behavior, which is often computationally infeasible for large-scale multimodal architectures \cite{jagielski2024students}. Even heuristic alternatives, such as cosine similarity-based attacks \cite{ko2023} or self-influence functions \cite{cohen2024membership}, either reduce detection accuracy or still demand substantial computational resources \cite{oh2023membership}. \textbf{(2) The “auditor’s dilemma”:} existing multimodal MIAs typically assume access to explicitly paired inputs, such as a face image with a name or a speaker’s voice with textual metadata, to query the target model \cite{ko2023}, creating a fundamental paradox because third-party auditors (e.g., regulators) usually do not possess, and should not handle, users’ sensitive biometric data. Even worse, submitting such sensitive inputs to an untrusted model can introduce new exposure risks, undermining the very purpose of the privacy audit \cite{hu2022m}.

To resolve the auditor's dilemma, a text-only, unimodal gray-box auditing setting is necessitated by two primary compliance scenarios: (1) \textbf{Regulatory Auditing:} Agencies auditing open-weight models for PII (e.g., under the EU AI Act) often have weight access but are legally prohibited from handling citizens' raw biometric data. (2) \textbf{Third-party Compliance:} External auditors may receive model weights yet remain barred from sensitive bimodal user databases by regulations like GDPR/HIPAA. UMID addresses these by eliminating the need for both computationally expensive shadow models and privacy-violating bimodal queries.

These challenges underscore an urgent need for auditing methods that avoid querying the target model with matched multimodal pairs \cite{liu2024multimodal}. In this paper, we reveal a key insight that resolves this dilemma: the presence of PII in the multimodal training set can be accurately audited using exclusively the text modality. By exploiting the cross-modal alignment behavior of these foundational encoders, we demonstrate that privacy-preserving, single-modality probing is sufficient to detect severe memorization. To this end, we study MIAs in a more realistic identity-level auditing setting. Specifically, we assume each training sample consists of a textual PII description $t_i$ (e.g., a name) and a corresponding non-textual modality $m_i$ (e.g., a face image for CLIP or a voice recording for CLAP). Given only a target text query $t$, the auditor seeks to determine whether there exists any training sample $(t_i, m_i)$ such that $t_i = t$. Crucially, this must be done under a strict \textbf{unimodal constraint}: even if the auditor possesses real biometric samples, they are prohibited from submitting them to the target model. This restriction enables lightweight, text-only queries that mitigate privacy leakage risk and are practical for third-party deployment (Figure \ref{fig:intro}), proving that auditors can effectively assess multimodal privacy risks without ever exposing sensitive biometric data.

\begin{figure*}[t]
    \centering
    \subfloat[\small Comparison of UMID and traditional MIA methods. Traditional approaches rely on costly shadow models or bimodal inputs (risking secondary privacy leakage), whereas UMID safely formulates membership inference as anomaly detection using text-only queries.\label{fig:intro}]{\includegraphics[width=0.48\textwidth]{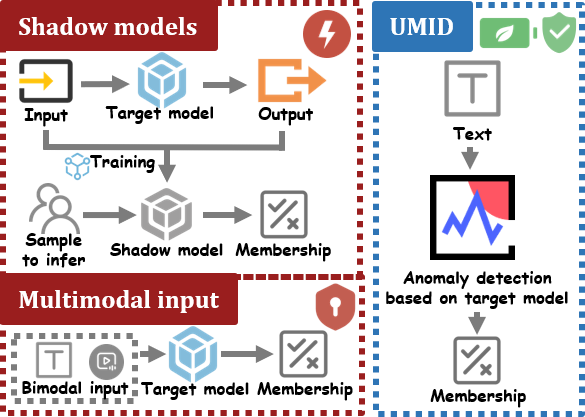}}
    \hfill
    \subfloat[\small Visualization of geometric separation. The extracted similarity and variability features via latent inversion exhibit a clear distributional gap between samples within and outside the training dataset of the target model.\label{fig:separation}]{\includegraphics[width=0.48\textwidth]{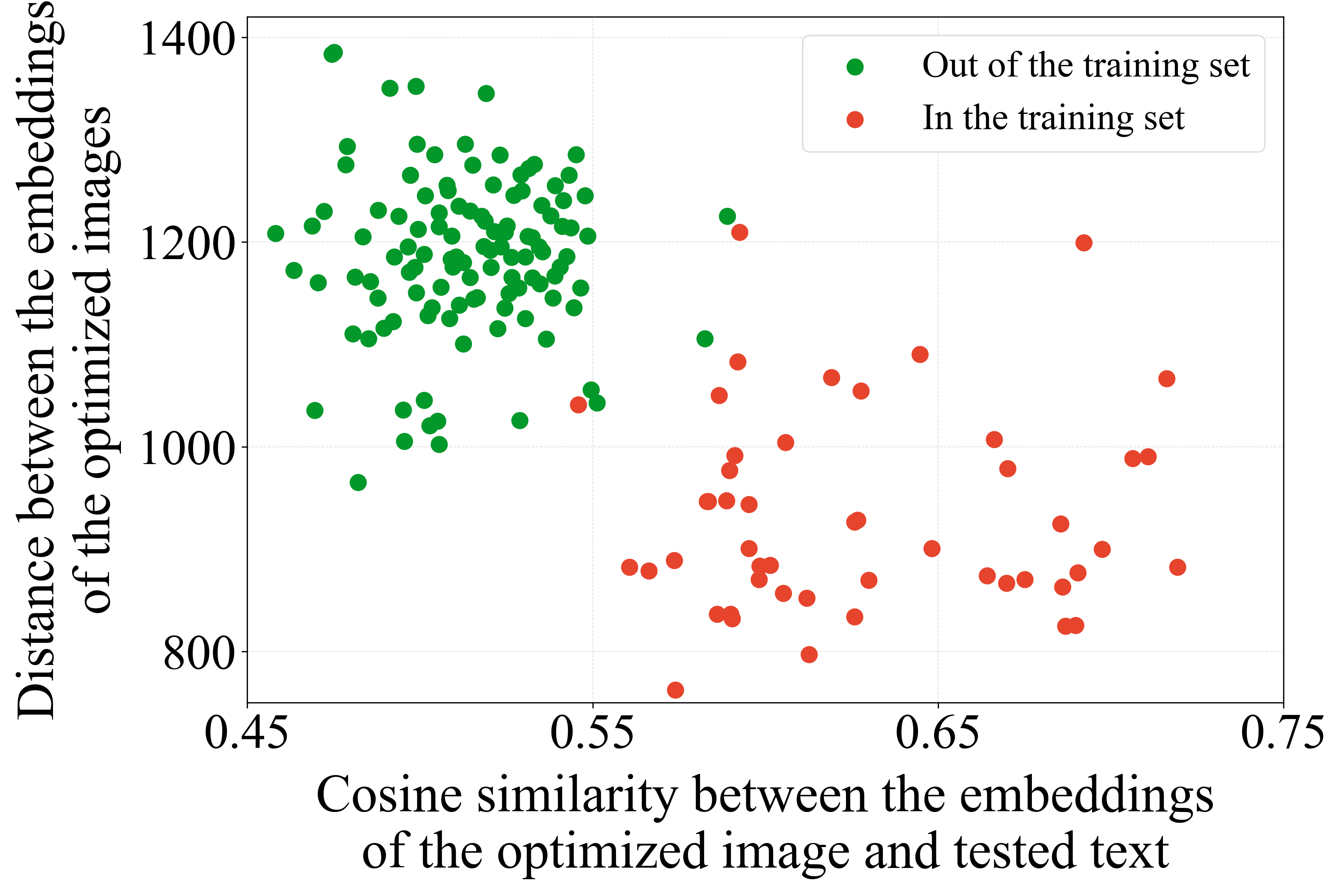}}
    
    \caption{Overview of the UMID auditing framework and the resulting distributional gap. The UMID method enables text-only membership inference (a) and achieves clear geometric separation between training and non-training samples (b).}
\end{figure*}

In this paper, we propose \textbf{U}nimodal \textbf{M}embership \textbf{I}nference \textbf{D}etector (UMID), a text-only query framework that detects PII memorization in contrastive models without exposing sensitive biometric inputs or training costly shadow models. Given a target text description, UMID performs cross-modal latent inversion with multiple randomized initializations and extracts two key statistics: (1) \textbf{similarity}, quantified as the average cosine alignment between the optimized embedding and the queried text embedding; and (2) \textbf{variability}, measured by the mean squared dispersion of the optimized embeddings across runs. UMID then compares these statistics against a lightweight non-member baseline built from synthetic textual “gibberish” and makes a membership decision using a voting ensemble of anomaly detectors. The intuition behind this design is straightforward yet effective: texts corresponding to identities seen during training (members) yield more accurate and consistent modal reconstructions under inversion than unseen texts (non-members). This discrepancy creates a distinct distributional separation that can be reliably detected, as illustrated in Figure \ref{fig:separation}.

\paragraph{Contributions} Our contributions are three-fold:
\begin{itemize}
	\item We introduce UMID, a membership inference detector for contrastive pretraining models that operates solely on textual PII queries. UMID eliminates costly shadow-model training and thus overcomes the computationally prohibitive costs faced by prior MIAs. Moreover, UMID avoids querying the target model with explicitly paired multimodal inputs, resolving the auditor’s dilemma and adhering to the principle of multimodal data protection.

	\item We provide both theoretical and empirical evidence showing that member and non-member texts induce well-separated feature distributions under text-guided cross-modal inversion. This enables membership inference via unsupervised anomaly detection using randomly generated textual gibberish.

	\item We conduct comprehensive experiments across diverse CLIP and CLAP architectures, demonstrating that UMID consistently outperforms existing MIA methods in both effectiveness and efficiency, even though it queries with text alone, rather than paired multimodal inputs.
\end{itemize}

%% file: 2_relatwork_revised.tex
\section{Related Work}

\paragraph{Contrastive pretraining models}
Multimodal contrastive pretraining \cite{radford2021learning} has become a core framework for learning unified representations across heterogeneous modalities, including images, audio, and natural language \cite{li2025closer}. Representative methods, such as CLIP \cite{khosla2020supervised} and CLAP \cite{wu2025collap}, train dual encoders with contrastive objectives on web-scale paired data to construct a shared embedding space, where semantically matched pairs are pulled together while mismatched pairs are pushed apart. Such an alignment in a shared embedding space enables strong zero-shot transfer, cross-modal retrieval, and a broad range of downstream applications without fine-tuning via task-specific supervision \cite{tsai2025construction}. Recent work further improves contrastive pretraining by strengthening cross-modal alignment and scalability \cite{cao2025supervised,hunt2025contrastive}, introducing codebook-based representations, and leveraging larger and more curated audio--text resources \cite{xie2025audiotime}. Despite these advances, the scale and opacity of the underlying training corpora raise growing concerns \cite{hu2022m}: sensitive information present in the data may be implicitly retained in learned embeddings, leading to privacy risks \cite{golatkar2022mixed}.

\paragraph{PII leakage in MLLMs}
The unprecedented capabilities of Multimodal Large Language Models (MLLMs) stem from extensive pretraining on massive and weakly supervised datasets \cite{ashqar2025advancing,cheng2025pbi,pang2026steering,xiao2025pixclip}. Such model training processes, however, inevitably introduces privacy risks as these web-scale datasets frequently contain personally identifiable information (PII), including faces, names, and voiceprints \cite{liu2025protecting,tran2025privacypreserving}. A growing number of studies demonstrate that multimodal encoders and associated LLM components can memorize identity-linked features and sensitive attributes, rendering them susceptible to various privacy attacks such as model stealing, knowledge extraction, data reconstruction, and membership inference  \cite{ko2023}. These risks are especially pronounced in  contrastive frameworks such as CLIP, where strong cross-modal alignment allows adversaries to elicit sensitive identity associations through crafted queries  \cite{kim2024}. Consequently, developing rigorous auditing protocols to detect such representation-level PII leakage remains a critical open problem \cite{pham2025never}.

\paragraph{Membership inference attacks}
Membership inference attacks (MIAs) serve as a fundamental and widely-used instrument for quantifying privacy leakage by ascertaining whether a specific data sample was included in a model's training set \cite{shokri2017membership}. Conventional MIA pipelines typically rely on the training of shadow models to approximate the target model's decision boundaries \cite{shokri2017membership} or the exploitation of loss-based confidence signals \cite{wang2025membership}. However, as architectures of modern large models usually scale to billions of parameters \cite{cheng2024reinforcement,duan2025oyster,cao2025agr,cheng2025inverse}, these strategies have become computationally prohibitive and often impractical \cite{9806361}. Recent efforts have adapted MIAs to contrastive models such as CLIP by leveraging cross-modal alignment behavior, including identity detection inference attacks (IDIA) \cite{hintersdorf2024does}, cosine similarity attacks (CSA) \cite{ko2023}, and weak supervision attacks (WSA) \cite{samira2025variance}. While these approaches successfully reduce or eliminate the need for shadow training, they paradoxically necessitate querying the target model with  biometric inputs (e.g., facial images or voice recordings) during inference \cite{tao2025range}. This highlights an urgent need for novel membership detectors that can operate effectively without involving sensitive bimodal PII.

%% file: 3_methodology_revised.tex
\section{Method}

\subsection{Problem Formulation}
\paragraph{Auditing scenario}
Consider a multimodal contrastive pretraining model $\mathcal{M}$ (e.g., CLIP \cite{khosla2020supervised} or CLAP \cite{wu2025collap}) pre-trained on a dataset $D_{\text{train}}$.
Each training sample $s_i=(t_i,m_i)$ represents a paired identity (i.e., PII) of an individual, where $t_i$ denotes a textual description (e.g., a person name) and $m_i$ denotes a corresponding non-textual modality (e.g., a face image for CLIP or a voice recording for CLAP). The objective of the detector is to perform \textit{identity-level membership inference}: given a target identity $t$, the detector aims to determine whether the PII sample $(t, m)$ was present in $\mathcal{D}_{\text{train}}$. Notably, this is conducted under a \textit{unimodal privacy constraint}: the auditor queries the model $\mathcal{M}$ using only textual descriptions. Even if the auditor possesses real biometric samples of the target, they are prohibited from submitting them to $\mathcal{M}$ to prevent secondary privacy leakage, which we term as \textit{the auditor's dilemma} in Section \ref{sec:intro}.

\paragraph{Detector's capability} We operate under a gray-box model auditing setting, where the detector is assumed to have query and gradient access to the frozen encoders, $\phi_{\text{text}}$ and $\phi_{\text{mod}}$ to perform the latent inversion described in Section \ref{sec:feature_extraction}.
Crucially, this assumption does not require access to the original training pipeline: the auditor remains completely oblivious to the private training dataset $\mathcal{D}_{\text{train}}$, the specific training hyperparameters, or the optimizer states.
This setting aligns with standard safety auditing protocols for open-weight foundation models, where model parameters are accessible but training data remains proprietary.

\begin{figure*}[t]
	\centering
	\includegraphics[width=\textwidth]{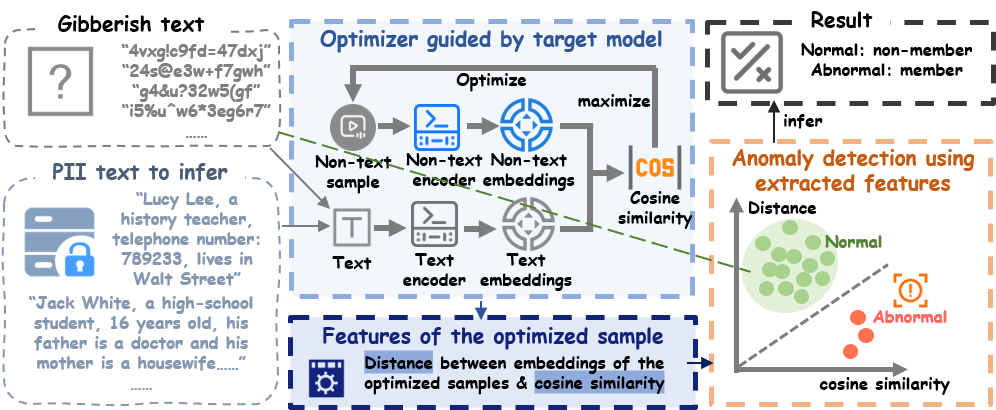}
	\caption{\small Pipeline of UMID. We employ an optimizer guided by target model to align non-text embeddings with PII text embeddings, maximizing their cosine similarity. By analyzing similarity and variability features of these optimized samples relative to the synthetic gibberish baseline, an anomaly detection system identifies abnormal patterns to infer the membership of the input text.}
	\label{fig:pipeline}
\end{figure*}

\subsection{UMID: Unimodal Membership Inference Detector}
We propose UMID, a framework that reformulates membership inference as an unsupervised anomaly detection task based on cross-modal latent inversion. As illustrated in Figure \ref{fig:pipeline}, UMID operates in three logical steps: (1) extracting membership signatures via randomized modality inversion; (2) grounding the separability of these signatures via geometric theory; and (3) performing membership inference via anomaly detection calibrated on a semantic-null baseline.

\subsubsection{Feature Extraction via Latent Inversion}
\label{sec:feature_extraction}

For a target textual identity $t$, UMID elicits its membership signature by optimizing a modality input $x$ to align with the text embedding $v_t = \phi_{\text{text}}(t)$. Specifically, we perform $n$ independent optimization runs. In each run $i \in \{1, \dots, n\}$, we initialize $x_0^{(i)}$ from random noise and iteratively update it via gradient ascent to maximize the cosine similarity $S(x, t) = \frac{\phi_{\text{mod}}(x)^\top v_t}{\|\phi_{\text{mod}}(x)\| \|v_t\|}$. Let $v^{(1)}(t), \dots, v^{(n)}(t)$ denote the final optimized embeddings. We extract two trajectory statistics as the feature representation $f(t)$:

\begin{enumerate}
	\item \textbf{Similarity ($S_n$):} Measured by the average final similarity across $n$ runs, reflecting the model's ability to reconstruct a modality counterpart for $t$:
	\begin{equation} \label{eq:bar-v-def}
		S_n(t) := v_t^\top \left( \frac{1}{n}\sum_{i=1}^n v^{(i)}(t) \right).
	\end{equation}
	\item \textbf{Variability ($D_n^2$):} Measured by the mean squared dispersion of the optimized embeddings, capturing the consistency of the optimization landscape:
\begin{equation} \label{eq:Dn-def}
\begin{split}
    D_n^2(t) := & \frac{1}{n}\sum_{i=1}^n \|v^{(i)}(t) - \bar{v}_n(t)\|_2^2, \\
    &\text{where } \bar{v}_n(t) = \frac{1}{n}\sum_{i=1}^n v^{(i)}(t).
\end{split}
\end{equation}
\end{enumerate}

Note that $S_n(t)$ and $D_n^2(t)$ serve as proxies for alignment and consistency, respectively. Specifically, a higher $S_n(t)$ indicates higher similarity, while a lower $D_n^2(t)$ signifies lower variability (higher stability). Intuitively, member identities tend to yield higher similarity and lower variability due to stronger training-induced alignment. We formalize this intuition in the following theoretical analysis.

\subsubsection{Geometric Separation with Theoretical Justification} \label{sec:theory}

We provide a theoretical guarantee that the statistics $(S_n, D_n^2)$ fundamentally separate members from non-members. We model the latent space using $K$ unit-norm prototypes $\{\mu_k\}_{k=1}^K$. The insights are motivated by distinct geometric behaviors: 
For a \textbf{non-member} $t_{\mathrm{out}}$, its embedding is \textit{isotropic} relative to the prototypes, implying that the optimization explores the prototype space uniformly ($p_k(t_{\mathrm{out}}) \approx 1/K$).
In contrast, a \textbf{member} $t_{\mathrm{in}}$ aligns with a specific prototype $\mu_{y^\star}$ with a \textit{margin} $\gamma$. This alignment forces the optimization to concentrate exponentially on $y^\star$ ($p_{y^\star}(t_{\mathrm{in}}) \approx 1$). These properties are well-grounded in contrastive learning mechanics: the training objective (e.g., InfoNCE) explicitly optimizes for high cosine similarity (margin) for members, whereas unseen non-members naturally exhibit high-dimensional isotropy.

Based on these properties, we define the population-level statistics $S_\infty(t)$ and $D_\infty^2(t)$, which capture the ideal behavior as $n \to \infty$. Members exhibit high alignment ($S_\infty \approx \gamma$) and low dispersion ($D_\infty^2 \approx 0$), whereas non-members exhibit low alignment ($S_\infty \approx 0$) and high dispersion ($D_\infty^2 \approx 1 - 1/K$). The following theorem provides a finite-sample guarantee for this separation rule.
\begin{theorem}[Finite-sample geometric separation]\label{thm:main}
	Given a member text $t_{\mathrm{in}}$ and a non-member text $t_{\mathrm{out}}$ satisfying the geometric properties above.
	Let $\Delta_S:=S_\infty(t_{\mathrm{in}})-S_\infty(t_{\mathrm{out}})$ and $\Delta_D:=D_\infty^2(t_{\mathrm{out}})-D_\infty^2(t_{\mathrm{in}})$ be the population gaps, and define $\Gamma:=\min\{\Delta_S,\Delta_D\}>0$.
	For sufficiently large embedding dimension $d=\Omega(\log(K/\delta))$ and number of runs $n=\Omega(\Gamma^{-2}\log(1/\delta))$, there exist thresholds $s_{\mathrm{thr}}$ and $d_{\mathrm{thr}}^2$ such that with probability at least $1-\delta$,
\begin{equation}
    \begin{aligned}
        S_n(t_{\mathrm{in}}) \ge s_{\mathrm{thr}}, \quad D_n^2(t_{\mathrm{in}}) \le d_{\mathrm{thr}}^2, \\
        S_n(t_{\mathrm{out}}) \le s_{\mathrm{thr}}, \quad D_n^2(t_{\mathrm{out}}) \ge d_{\mathrm{thr}}^2.
    \end{aligned}
\end{equation}
\end{theorem}

Theorem \ref{thm:main} formally validates that our randomized reverse-optimization operationalizes the geometric distinction described above. It guarantees that the empirical statistics $(S_n, D_n^2)$ inherit the population-level separation with high probability. Furthermore, it quantifies the sample complexity, highlighting that the computational cost scales inversely with the square of the separation margin $\Gamma$. Detailed proofs are deferred to \ref{sec:appendix_proof}.

\subsubsection{Membership Inference via Anomaly Detection}\label{sec:anomaly detection}
Translating the theoretical separation into a practical algorithm requires estimating the decision boundaries. As the theoretical thresholds are unknown, we reformulate the problem as unsupervised anomaly detection.

\paragraph{Semantic-null baseline} To approximate the non-member distribution, we generate a reference set of $\ell$ \textit{gibberish strings} $\mathcal{G} = \{g_1, \ldots, g_{\ell}\}$. Since $\mathcal{G}$ contains no semantic information and did not appear in training, it satisfies the isotropy condition stated in Section \ref{sec:theory} and serves as a typical proxy for the null hypothesis.

\paragraph{Inference procedure} We extract features $\mathcal{F}_{\mathcal{G}}$ from $\mathcal{G}$ to train an ensemble of $k_{\text{det}}$ anomaly detectors (e.g., Isolation Forests). Each detector learns a boundary encompassing the non-member region. For a test identity $t$, UMID classifies it as a member if the majority of detectors flag its features $f(t)$ as an anomaly, indicating a significant shift toward the high-similarity and low-variability characteristic of members.

\paragraph{Algorithm of UMID} The complete algorithm, including feature extraction, geometric separation, and membership inference, is stated in Algorithm \ref{alg:umid}.

\subsubsection{Optional Enhancement with Local Modality Samples}
In scenarios where the auditor possesses local real modality samples (which \textbf{cannot} be submitted to $\mathcal{M}$), UMID can include an auxiliary coherence feature. We use an external feature extractor $F$ (e.g., DeepFace~\cite{taigman2014deepface}) to compute $R(t)$: the average pairwise $\ell_2$ distance between the local samples and the optimized embeddings provided by UMID. Intuitively, members yield smaller $R(t)$ as optimized samples converge to true identity features. We incorporate $R(t)$ by performing K-means clustering ($K=2$) on the augmented feature set $\{(S_n, D_n^2, R)\}$, adding the cluster assignment as an additional vote to the ensemble.

\begin{algorithm}[t]
	\caption{The algorithm of UMID.}
	\label{alg:umid}
	\textbf{Input}: Model $\mathcal{M}$; Target text $t$; number of runs $n$; iterations per run $m$; learning rate $\eta$. \\
	\textbf{Output}: Membership decision (member/non-member).
	\begin{algorithmic}[1]
		\State \textbf{Phase 1: Baseline construction (Offline)}
		\State Generate gibberish set $\mathcal{G} = \{g_1, \dots, g_\ell\}$
		\For{$g_k \in \mathcal{G}$}
		\State $S_n(g_k), D_n^2(g_k) \gets \textsc{LatentInversion}(\mathcal{M}, g_k, n, m, \eta)$
		\EndFor
		\State Train ensemble detectors $\mathcal{E} = \{D_1, \dots, D_k\}$ on features $\{(S_n(g_k), D_n^2(g_k))\}$
		
		\State \textbf{Phase 2: Membership inference (Online)}
		\State $S_n(t), D_n^2(t) \gets \textsc{LatentInversion}(\mathcal{M}, t, n, m, \eta)$
		\State Votes $\gets \sum_{j=1}^k \mathbb{I}(D_j(S_n(t), D_n^2(t)) = \text{Anomaly})$
		\State \textbf{Return} Member if Votes $> N$, else Non-member.
		
		\Statex \hrulefill
		\Function{LatentInversion}{$\mathcal{M}, t, n, m, \eta$}
		\State $v_t \gets \phi_{\text{text}}(t)$; \quad $\mathcal{V} \gets \emptyset$  \Comment{Set of optimized embeddings}
		\For{$i = 1$ \textbf{to} $n$}
		\State $x_0 \sim \mathcal{N}(0, I)$ \Comment{Random initialization}
		\For{$j = 0$ \textbf{to} $m-1$}
		\State $g \gets g_j \gets \nabla_{x_j} \left(\frac{v_t^\top \phi_{\text{mod}}(x_j)}{\|v_t\| \|\phi_{\text{mod}}(x_j)\|}\right)$
		\State $x_{j+1} \gets x_j + \eta \cdot g$
		\EndFor
		\State $\mathcal{V} \gets \mathcal{V} \cup \{\phi_{\text{mod}}(x_m)\}$
		\EndFor
		\State Compute $S_n, D_n^2$ via Eqn. \eqref{eq:bar-v-def} and Eqn. \eqref{eq:Dn-def}
		\State \textbf{Return} $S_n, D_n^2$
		\EndFunction
	\end{algorithmic}
\end{algorithm}

%% file: 4_evaluation_revised.tex
\section{Experiments}
\label{sec:evalu}

We present a comprehensive evaluation of UMID across two distinct multimodal contrastive frameworks: CLIP (vision-text) and CLAP (audio-text). The experiments aim to validate UMID's superiority in both detection effectiveness and computational efficiency compared to state-of-the-art baselines.

\subsection{Experimental Setup}

\subsubsection{Datasets}
\paragraph{Image-text (CLIP)}

We construct the \textbf{CelebA} dataset following previous work \cite{hintersdorf2024does}. We integrate image-text pairs from \textbf{FaceScrub} \cite{kemelmacher2016megaface} and \textbf{LAION-5B} \cite{schuhmann2022laion} into the \textbf{CC3M} \cite{changpinyo2021conceptual} backbone to form the training set, ensuring a gender-balanced distribution across 200 selected identities. Specifically, we set a threshold to exclude individuals with excessively high frequencies in LAION-400M, maintaining 100 members and 100 non-members to ensure experimental reliability. The training set comprises 100 celebrities (members) and 100 hold-out identities (non-members). To simulate varying degrees of memorization, we curate two subsets: \textit{One-Shot} (1 photo/person) and \textit{Many-Shot} (75 photos/person).

\paragraph{Audio-text (CLAP)}
We utilize \textbf{LibriSpeech}~\cite{7178964} and construct a richer speaker recognition dataset based on \textbf{CommonVoice18.0}~\cite{DBLP:journals/corr/abs-1912-06670}. The latter involves 3,000 speakers (1,500 members/1,500 non-members) with detailed PII descriptions (ID, age, gender) augmented by GPT-4o to generate semantic background narratives. Similarly, we evaluate under \textbf{One-Shot} (1 audio clip/user) and \textbf{Many-Shot} (50 audio clips/user) settings.


\subsubsection{Models and Baselines}
\paragraph{Models} For CLIP, we conduct membership inference on six target models spanning three visual backbones: ResNet-50, ResNet-50x4, and ViT-B/32. The ResNet-based models follow the standard ResNet architecture~\cite{he2016deep,theckedath2020detecting}, while ViT-B/32 adopts the Vision Transformer design, covering both CNN- and Transformer-based encoders. For CLAP, we adopt a dual-encoder architecture in which the audio encoder is HTSAT~\cite{chen2022hts}, a transformer composed of four groups of Swin Transformer blocks~\cite{liu2021swin}, and the text encoder is RoBERTa~\cite{liu1907roberta}. The penultimate-layer outputs of both encoders (768 dimensions) are projected to a shared 512-dimensional embedding space via a two-layer MLP with ReLU activation. All target models are trained from scratch on their respective member datasets to establish unambiguous ground-truth membership labels. When real facial images are available for optional enhancement, DeepFace~\cite{serengil2020lightface} is used as an external feature extractor.
Following standard MIA protocols \cite{carlini2022membership}, we train models from scratch to ensure absolute ground-truth certainty. Auditing public foundation models (e.g., LAION-trained CLIP) is methodologically unreliable due to massive data contamination and unindexed training sets, which obscure membership labels. Our controlled training on industry-standard architectures (ResNet, ViT, RoBERTa) ensures that evaluation metrics reflect true algorithmic capability rather than dataset noise.

\paragraph{Baselines} We compare UMID against a spectrum of MIA methods:
\begin{itemize}
	\item \textbf{Shadow-model-based:} Standard MIA~\cite{shokri2017membership} is the base membership inference attack using shadow models; Audio Auditor~\cite{miao2021audio} trains shadow models and extracts audio features for inference; and AuditMI~\cite{teixeira2025exploring} trains shadow model using input utterances and features from model outputs; SLMIA-SR~\cite{SLMIA-SR} employs a shadow speaker recognition system to train attack model. These require training proxy models to mimic target behavior.
	\item \textbf{Metric-based:} IDIA~\cite{hintersdorf2024does} detects training membership by verifying if a model correctly predicts a target's identity from candidate prompts across multiple reference images; WSA~\cite{ko2023} infers membership via image-text cosine similarity, enhanced by a weakly supervised MIA framework trained on post-release non-member data; C-WSA~\cite{samira2025variance} treats low-variance samples as pseudo-members to train a robust image-feature classifier via confidence-based weak supervision. These rely on thresholding signals like loss or cosine similarity, often requiring real modality inputs (e.g., images/audio).
\end{itemize}

\begin{table}[t]
	\centering
	\caption{Comparison of CLIP (image-text) methods. Best and second-best results are \textbf{bolded}.}
	\label{table:clipbaseline}
	\renewcommand{\arraystretch}{0.95}
	\resizebox{\linewidth}{!}{
		\begin{tabular}{@{}ccc|ccc|c@{}}
			\toprule
			\multirow{2}[2]{*}{\textbf{Setting}} & \multirow{2}[2]{*}{\textbf{Method}} & \multirow{2}[2]{*}{\textbf{Real data}} & \multicolumn{3}{c|}{\textbf{Effectiveness}} & \textbf{Efficiency} \\
			\cmidrule{4-7} 
			& & & \textbf{Precision} & \textbf{Recall} & \textbf{Accuracy} & \textbf{Time} \\
			\midrule
			\multirow{6}[4]{*}{\shortstack{ResNet-50\\(one-shot)}} 
			& MIA~\cite{shokri2017membership} & 1 img & .7251 & .6873 & .7042 & 4.32h \\
			& IDIA~\cite{hintersdorf2024does} & 3 imgs & .6922 & .4031 & .6836 & 0.845s \\
			& WSA~\cite{ko2023} & 1 img & .6653 & .2925 & .6675 & 2158.1s \\
			& C-WSA~\cite{samira2025variance} & 1 img & .7210 & .4530 & .7125 & 2141.4s \\
			& \textbf{UMID (Ours)} & Null & \textbf{.8634} & \textbf{.9821} & \textbf{.9172} & \textbf{0.628s} \\
			& \textbf{UMID (Ours)} & 1 img & \textbf{.9145} & \textbf{.9912} & \textbf{.9528} & \textbf{0.745s} \\
			\midrule
			\multirow{6}[4]{*}{\shortstack{ResNet-50\\(many-shot)}} 
			& MIA~\cite{shokri2017membership} & 1 img & .7530 & .7086 & .7219 & 4.32h \\
			& IDIA~\cite{hintersdorf2024does} & 3 imgs & .6901 & .3998 & .6907 & 0.851s \\
			& WSA~\cite{ko2023} & 1 img & .6625 & .2867 & .6710 & 2143.7s \\
			& C-WSA~\cite{samira2025variance} & 1 img & .7245 & .4605 & .7218 & 2149.5s \\
			& \textbf{UMID (Ours)} & Null & \textbf{.8642} & \textbf{.9835} & \textbf{.9031} & \textbf{0.634s} \\
			& \textbf{UMID (Ours)} & 1 img & \textbf{.9168} & \textbf{.9925} & \textbf{.9554} & \textbf{0.752s} \\
			\midrule
			\multirow{6}[4]{*}{\shortstack{ResNet-50x4\\(one-shot)}} 
			& MIA~\cite{shokri2017membership} & 1 img & .7937 & .6492 & .7218 & 4.37h \\
			& IDIA~\cite{hintersdorf2024does} & 3 imgs & .6625 & .3980 & .6957 & 0.851s \\
			& WSA~\cite{ko2023} & 1 img & .6712 & .2912 & .6808 & 2144.3s \\
			& C-WSA~\cite{samira2025variance} & 1 img & .7305 & .4592 & .7234 & 2148.6s \\
			& \textbf{UMID (Ours)} & Null & \textbf{.8613} & \textbf{.9747} & \textbf{.9355} & \textbf{0.633s} \\
			& \textbf{UMID (Ours)} & 1 img & \textbf{.9128} & \textbf{.9896} & \textbf{.9535} & \textbf{0.685s} \\
			\midrule
			\multirow{6}[4]{*}{\shortstack{ResNet-50x4\\(many-shot)}} 
			& MIA~\cite{shokri2017membership} & 1 img & .7654 & .6981 & .7189 & 4.36h \\
			& IDIA~\cite{hintersdorf2024does} & 3 imgs & .7085 & .3904 & .7167 & 0.844s \\
			& WSA~\cite{ko2023} & 1 img & .6724 & .2935 & .6685 & 2156.7s \\
			& C-WSA~\cite{samira2025variance} & 1 img & .7364 & .4670 & .7321 & 2150.3s \\
			& \textbf{UMID (Ours)} & Null & \textbf{.8712} & \textbf{.9916} & \textbf{.9462} & \textbf{0.637s} \\
			& \textbf{UMID (Ours)} & 1 img & \textbf{.9156} & \textbf{.9908} & \textbf{.9562} & \textbf{0.691s} \\
			\midrule
			\multirow{6}[4]{*}{\shortstack{ViT-B/32\\(one-shot)}} 
			& MIA~\cite{shokri2017membership} & 1 img & .6923 & .5987 & .6694 & 4.35h \\
			& IDIA~\cite{hintersdorf2024does} & 3 imgs & .6783 & .3746 & .6772 & 1.245s \\
			& WSA~\cite{ko2023} & 1 img & .6323 & .2964 & .6812 & 2143.1s \\
			& C-WSA~\cite{samira2025variance} & 1 img & .6945 & .4230 & .6912 & 2146.2s \\
			& \textbf{UMID (Ours)} & Null & \textbf{.7091} & \textbf{.6385} & \textbf{.6837} & \textbf{0.667s} \\
			& \textbf{UMID (Ours)} & 1 img & \textbf{.7462} & \textbf{.6815} & \textbf{.7145} & \textbf{0.701s} \\
			\midrule
			\multirow{6}[4]{*}{\shortstack{ViT-B/32\\(many-shot)}} 
			& MIA~\cite{shokri2017membership} & 1 img & .6958 & .6341 & .6732 & 4.35h \\
			& IDIA~\cite{hintersdorf2024does} & 3 imgs & .6890 & .3811 & .6927 & 1.251s \\
			& WSA~\cite{ko2023} & 1 img & .7045 & .2806 & .6895 & 2144.8s \\
			& C-WSA~\cite{samira2025variance} & 1 img & .7012 & .4305 & .7008 & 2147.9s \\
			& \textbf{UMID (Ours)} & Null & \textbf{.7182} & \textbf{.6372} & \textbf{.6947} & \textbf{0.652s} \\
			& \textbf{UMID (Ours)} & 1 img & \textbf{.7498} & \textbf{.6844} & \textbf{.7186} & \textbf{0.712s} \\
			\bottomrule
		\end{tabular}
	}
\end{table}

\begin{table}[t]
	\centering
	\caption{Comparison of CLAP (audio-text) methods. Best and second-best results are \textbf{bolded}.}
	\label{table:clap_comprehensive}
	\renewcommand{\arraystretch}{0.95}
	\resizebox{\linewidth}{!}{
		\begin{tabular}{@{}ccc|ccc|c@{}}
			\toprule
			\multirow{2}[2]{*}{\textbf{Setting}} & \multirow{2}[2]{*}{\textbf{Method}} & \multirow{2}[2]{*}{\textbf{Real data}} & \multicolumn{3}{c|}{\textbf{Effectiveness}} & \textbf{Efficiency} \\
			\cmidrule{4-7} 
			& & & \textbf{Precision} & \textbf{Recall} & \textbf{Accuracy} & \textbf{Time} \\
			\midrule
			\multirow{5}[4]{*}{\shortstack{LibriSpeech\\(one-shot)}} 
			& Audio Auditor~\cite{miao2021audio} & 1 audio & .6338 & .7324 & .6519 & 1.273s \\
			& SLMIA-SR~\cite{SLMIA-SR} & 1 audio & .7521 & .8864 & .8342 & 1.528s \\
			& AuditMI~\cite{teixeira2025exploring} & 1 audio & .8257 & .9526 & .8791 & 13.578s \\
			& \textbf{UMID (Ours)} & Null & \textbf{.8649} & \textbf{.9649} & \textbf{.9127} & \textbf{0.603s} \\
			& \textbf{UMID (Ours)} & 1 audio & \textbf{.8921} & \textbf{.9868} & \textbf{.9354} & \textbf{0.647s} \\
			\midrule
			\multirow{5}[4]{*}{\shortstack{LibriSpeech\\(many-shot)}} 
			& Audio Auditor~\cite{miao2021audio} & 1 audio & .6559 & .8013 & .6659 & 1.245s \\
			& SLMIA-SR~\cite{SLMIA-SR} & 1 audio & .7619 & .9007 & .8433 & 1.547s \\
			& AuditMI~\cite{teixeira2025exploring} & 1 audio & .8341 & .9804 & .8816 & 13.511s \\
			& \textbf{UMID (Ours)} & Null & \textbf{.8812} & \textbf{.9876} & \textbf{.9307} & \textbf{0.612s} \\
			& \textbf{UMID (Ours)} & 1 audio & \textbf{91.63} & \textbf{99.57} & \textbf{95.24} & \textbf{0.659s} \\
			\midrule
			\multirow{5}[4]{*}{\shortstack{CommonVoice\\(one-shot)}} 
			& Audio Auditor~\cite{miao2021audio} & 1 audio & .5485 & .6822 & .6052 & 1.347s \\
			& SLMIA-SR~\cite{SLMIA-SR} & 1 audio & .6539 & .7691 & .7048 & 1.597s \\
			& AuditMI~\cite{teixeira2025exploring} & 1 audio & .7143 & .8145 & .7436 & 15.214s \\
			& \textbf{UMID (Ours)} & Null & \textbf{.7496} & \textbf{.8601} & \textbf{.8179} & \textbf{0.721s} \\
			& \textbf{UMID (Ours)} & 1 audio & \textbf{.7602} & \textbf{.8955} & \textbf{.8356} & \textbf{0.753s} \\
			\midrule
			\multirow{5}[4]{*}{\shortstack{CommonVoice\\(many-shot)}} 
			& Audio Auditor~\cite{miao2021audio} & 1 audio & .5611 & .7358 & .6135 & 1.391s \\
			& SLMIA-SR~\cite{SLMIA-SR} & 1 audio & .6628 & .7927 & .7218 & 1.652s \\
			& AuditMI~\cite{teixeira2025exploring} & 1 audio & .7352 & .8481 & .7564 & 15.931s \\
			& \textbf{UMID (Ours)} & Null & \textbf{.7647} & \textbf{.8946} & \textbf{.8233} & \textbf{0.815s} \\
			& \textbf{UMID (Ours)} & 1 audio & \textbf{.7934} & \textbf{.9113} & \textbf{.8569} & \textbf{0.776s} \\
			\bottomrule
		\end{tabular}
	}
\end{table}

\subsubsection{Implementation and Evaluation}
All experiments are conducted on four NVIDIA GeForce RTX 3090 GPUs. Each setting is repeated five times, with the mean and standard deviation reported. UMID performs optimization in latent inversion for $n=100$ epochs, each consisting of $m=1000$ gradient descent iterations with a learning rate of $\eta = 3 \times 10^{-2}$. The anomaly detection ensemble includes LocalOutlierFactor~\cite{cheng2019outlier}, IsolationForest~\cite{liu2008isolation}, OneClassSVM~\cite{li2003improving,khan2014one}, and AutoEncoder~\cite{chen2018autoencoder}. These detectors are trained using $\ell=100$ synthetic gibberish strings generated by prompting GPT-3.5-turbo (see~\ref{sec:gibberish} for examples), with a default voting threshold of $N=3$. When real non-textual samples are available, a $K$-means model ($K=2$) is additionally incorporated, and the threshold is adjusted to $N'=4$.


Additionally, we evaluate effectiveness using \textbf{Precision}, \textbf{Recall}, and \textbf{Accuracy}. Efficiency is assessed via \textbf{Time}, measured as the wall-clock latency per query.

\subsection{Comparison with State-of-the-Art Methods}
\paragraph{Performance on vision-language models}
Table \ref{table:clipbaseline} provides a comprehensive comparison of UMID against prior membership inference baselines across CLIP backbones, data exposure regimes, and query modalities. On ResNet-based CLIP models (ResNet-50 and ResNet-50x4), UMID exhibits a clear and consistent advantage, achieving over $90\%$ accuracy using text-only (without real data) queries while driving recall to $97$–$99\%$. This represents a substantial improvement over the strongest baseline C-WSA, with gains of up to $+18$ accuracy points and, more critically, over $+50$ recall point. This suggests that metric- and shadow-model-based attacks incur severe false negatives and fail to capture subtle memorization effects in contrastive models. Additionally, UMID still maintains a consistent superior effectiveness over all baselines, despite that Vision Transformers (ViT-B/32) exhibit weaker text-only signals due to stronger generalization. This indicates its ability to effectively probe the more resilient Transformer latent space. Importantly, the effectiveness gains across all settings come without increased computational cost: UMID completes inference in sub-second time ($\approx 0.6$–$0.8$s), delivering orders-of-magnitude speedups over shadow-model approaches while eliminating reliance on real biometric samples. 
To mirror real-world auditing where non-members predominate, we evaluate UMID on an imbalanced CelebA split (10\% members). UMID achieves a 0.941 AUROC with a low 3.2\% FPR (at 90\% TPR). This demonstrates that our semantic-null baseline effectively characterizes the non-member distribution, preventing excessive false positives in skewed real-world scenarios.

\paragraph{Performance on audio-language models} Table~\ref{table:clap_comprehensive} compares CLAP membership auditing methods across two datasets, revealing three consistent findings: (i) Effectiveness and robustness: UMID forms the clear top tier in all settings, uniquely maintaining high precision and high recall simultaneously. On LibriSpeech, UMID already performs strongly without any real target audio (UMID-Null: $91.27\%/93.07\%$ accuracy in one-/many-shot), exceeding the strongest baseline AuditMI by a clear margin. Notably, with only one audio query, the gap further widens (up to $+7.08$ accuracy points) and recall approaches saturation. This advantage is amplified on the more challenging CommonVoice dataset, where all baselines degrade markedly (e.g., AuditMI $\approx 75\%$ accuracy) but UMID remains robust, widening both accuracy and recall gaps, indicating better generalization beyond clean speech. (ii) Exposure sensitivity. UMID’s one-shot performance is already close to many-shot (accuracy gaps $\leq 2.2$ points), suggesting that membership leakage is detectable even under minimal exposure when cross-modal evidence is properly exploited. (iii) Efficiency and deployability. In addition to its strong effectiveness, UMID runs in sub-second time ($\approx 0.6$--$0.8$s), offering an order-of-magnitude speedup over other baselines.

\begin{figure*}[tp] 
    \centering
    \subfloat[]{\includegraphics[width=0.32\textwidth]{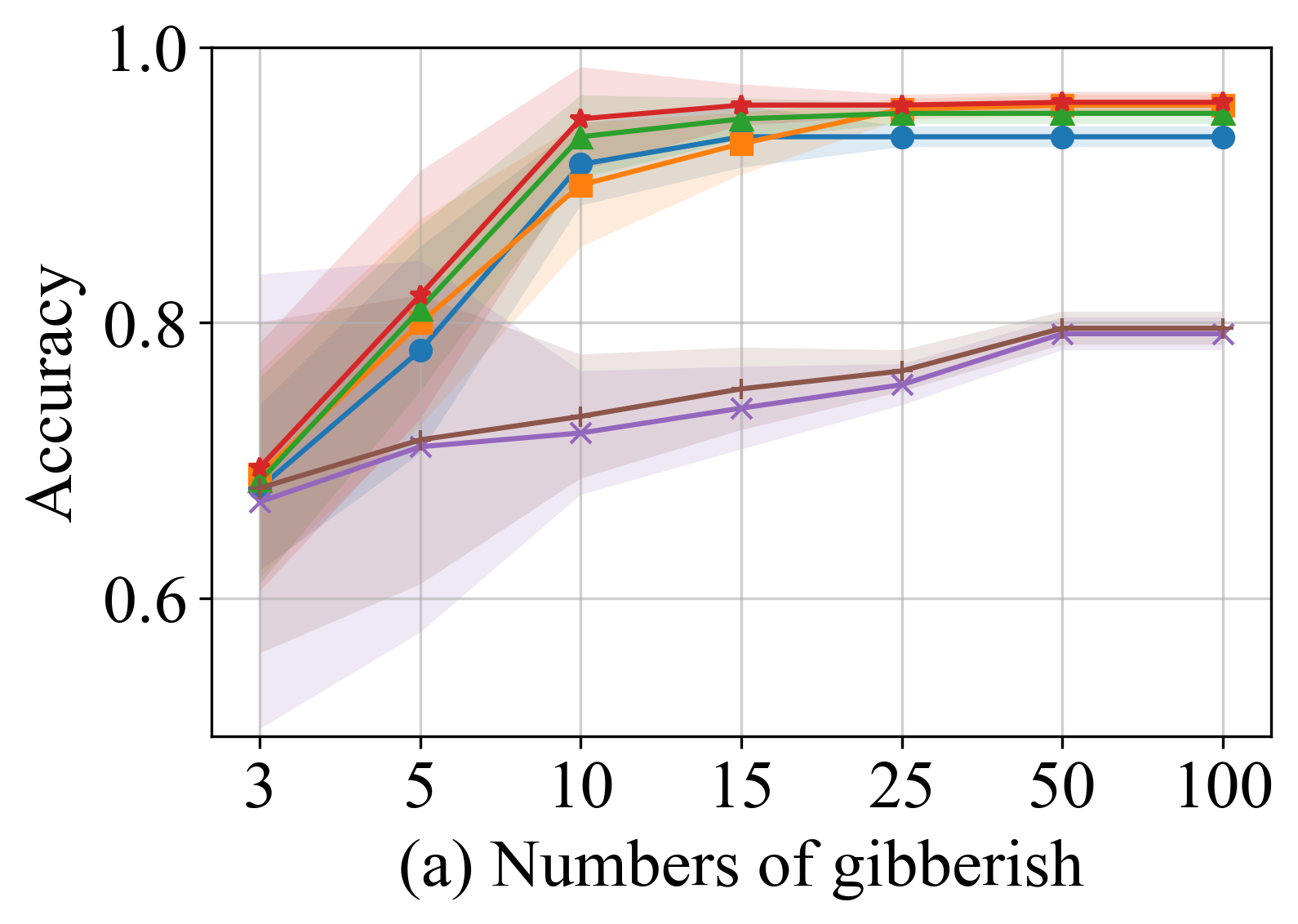}}
    \hfill
    \subfloat[]{\includegraphics[width=0.32\textwidth]{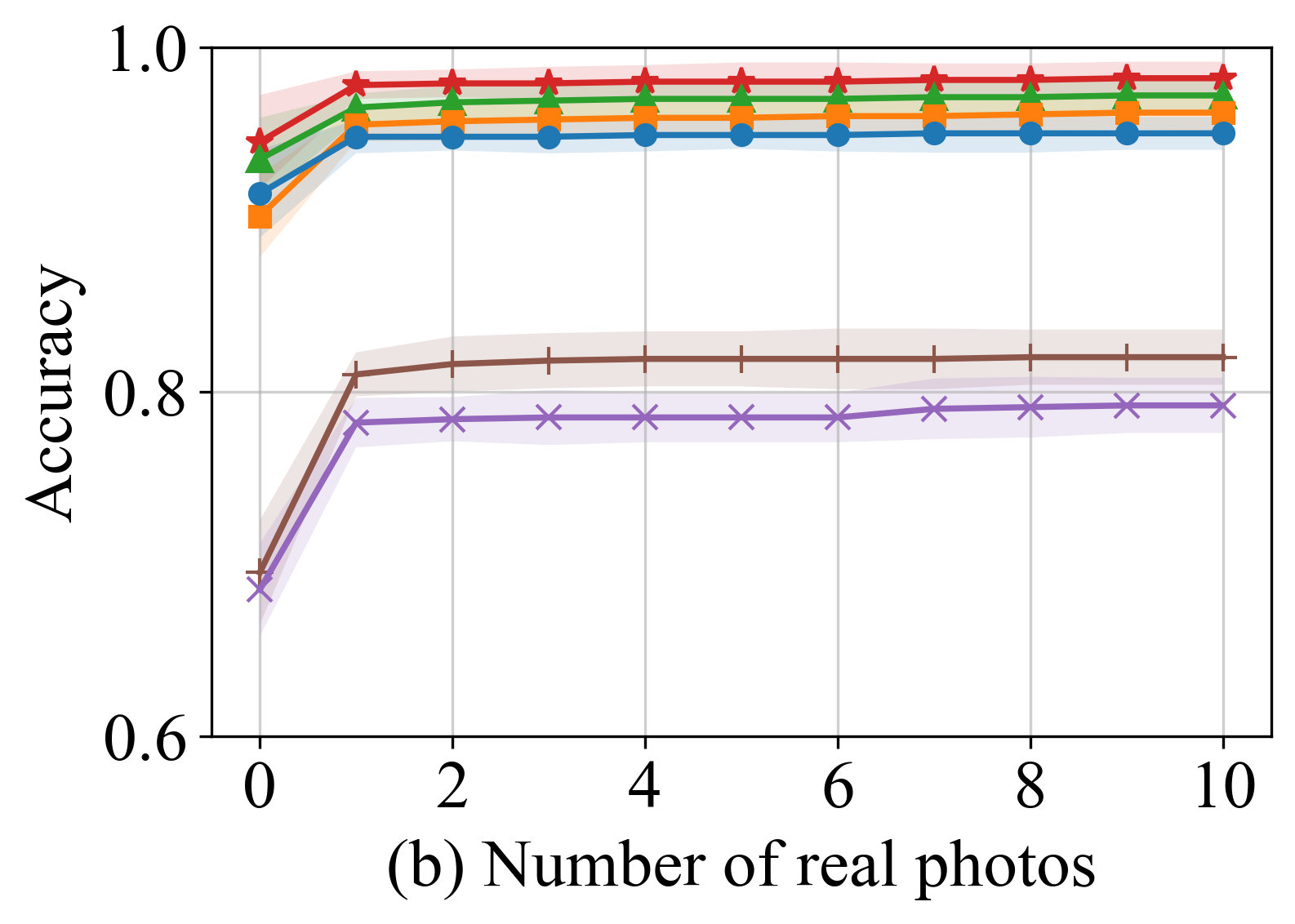}}
    \hfill
    \subfloat[]{\includegraphics[width=0.32\textwidth]{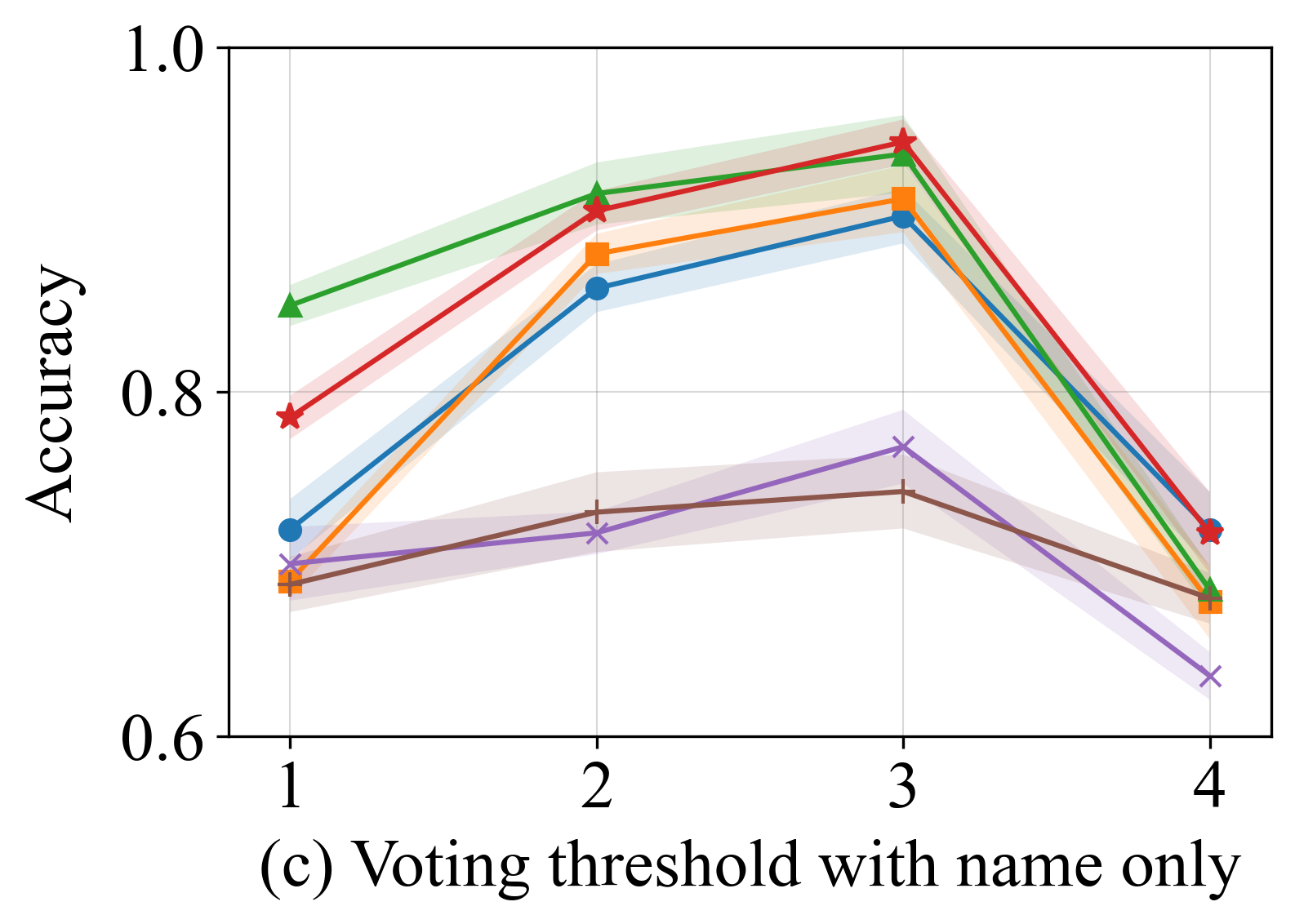}}
    \\[2mm]
    \subfloat[]{\includegraphics[width=0.32\textwidth]{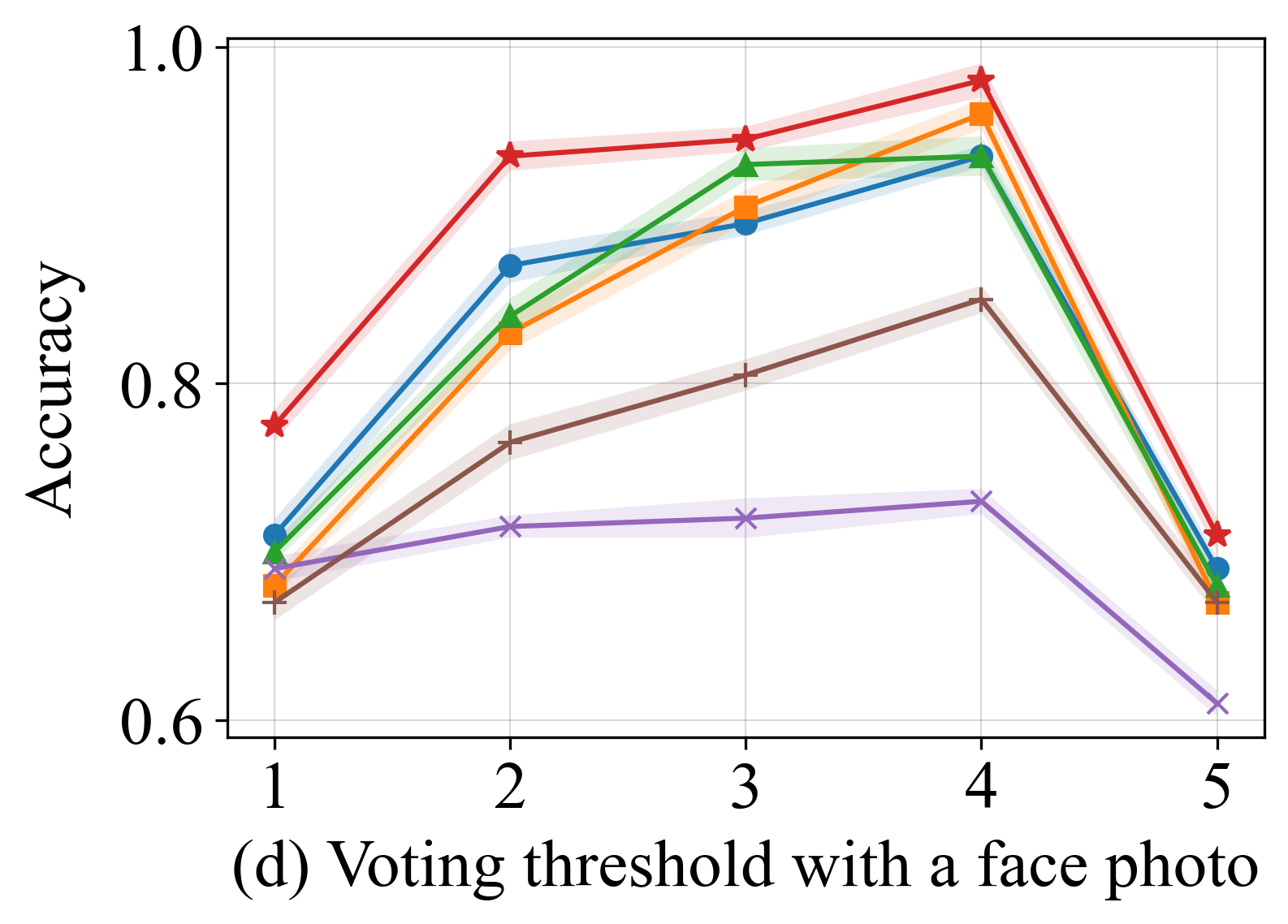}}
    \hfill
    \subfloat[]{\includegraphics[width=0.32\textwidth]{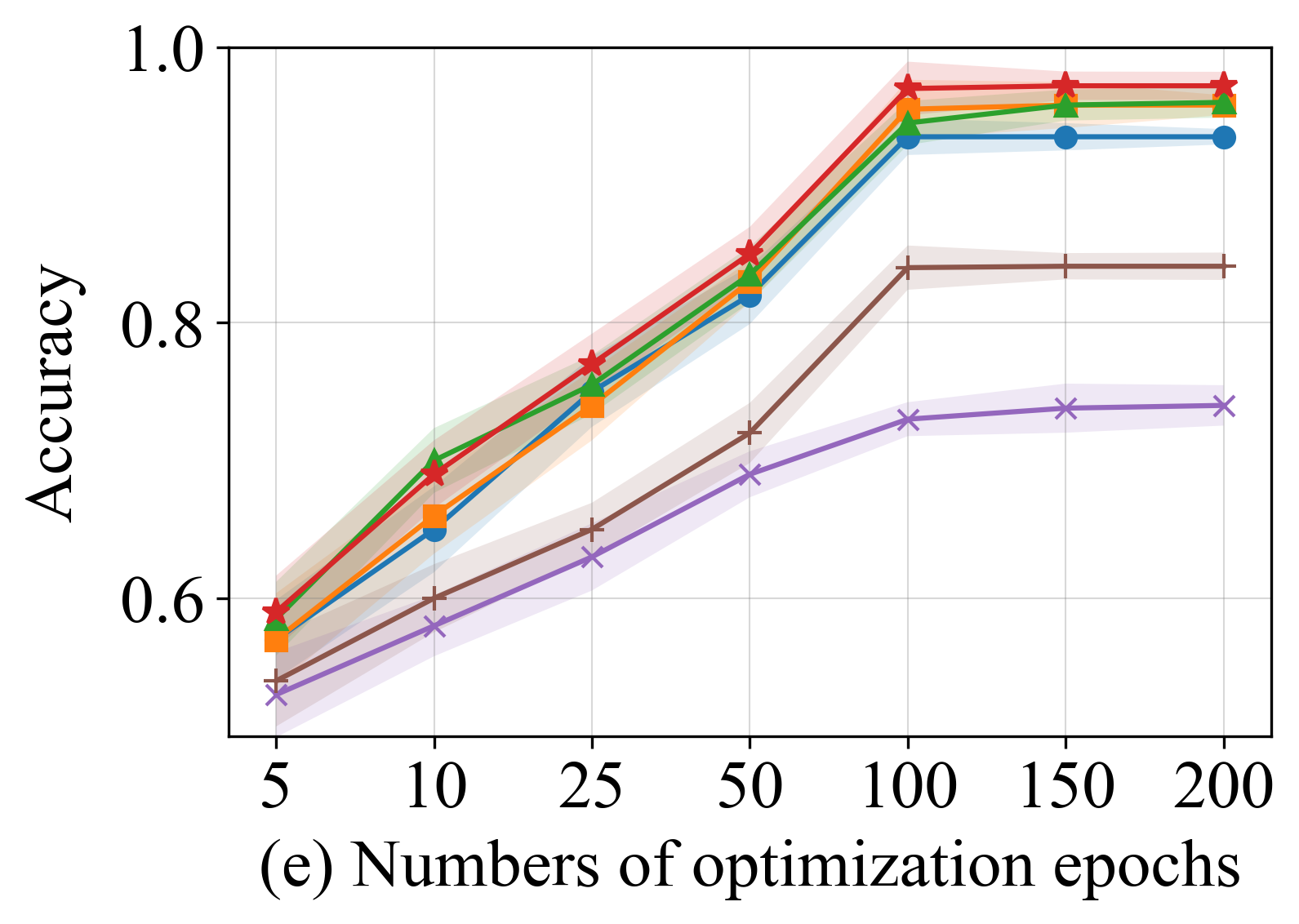}}
    \hfill
    \subfloat[]{\includegraphics[width=0.32\textwidth]{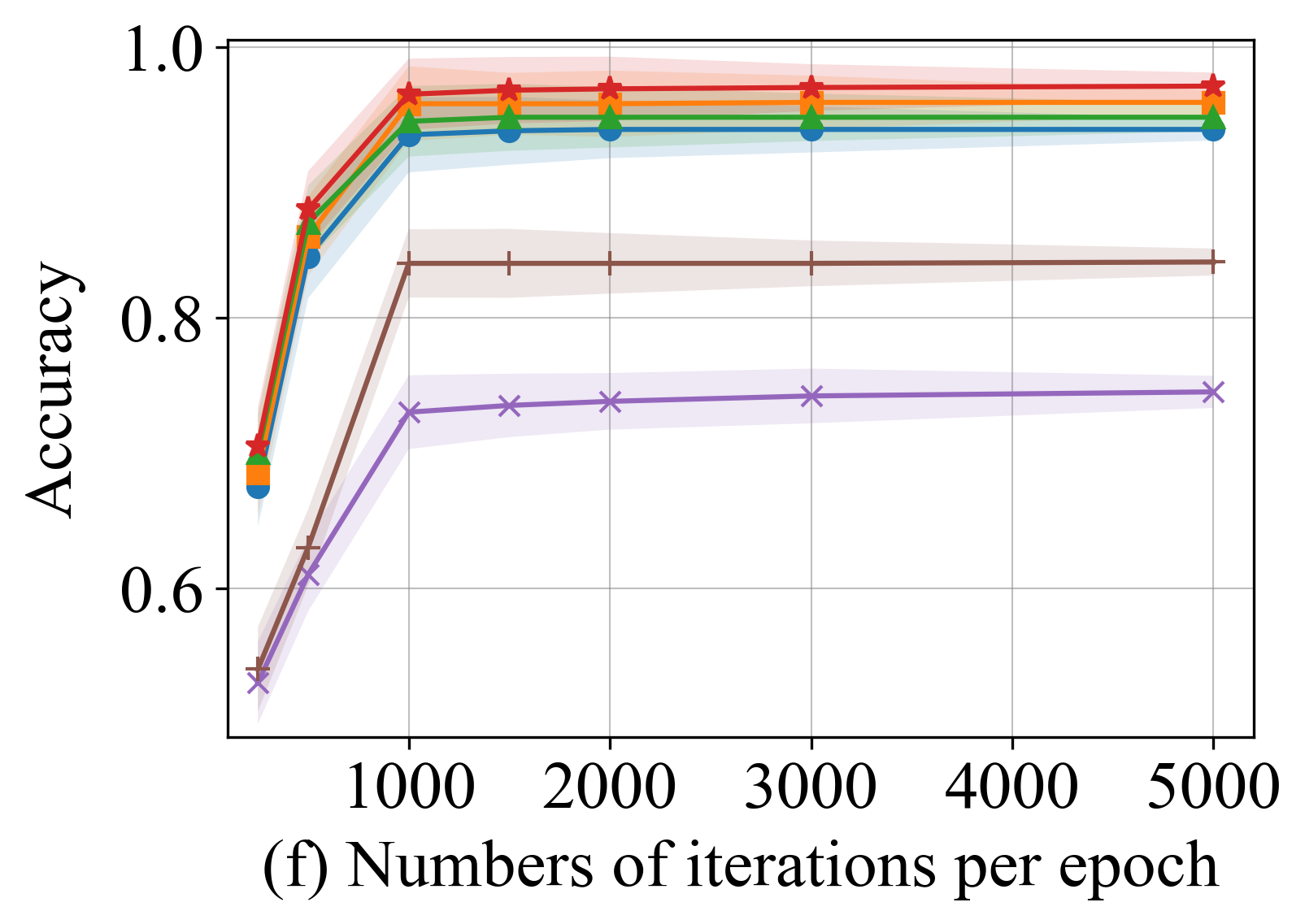}}
    
    \vspace{2mm}
    \includegraphics[width=0.5\textwidth]{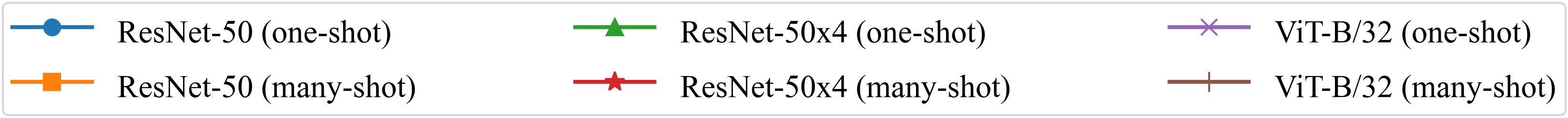}
    \caption{Detection accuracy for CLIP model (ResNet-50) under various parameters.}
    \label{fig:clipablation}
\end{figure*}

\begin{figure*}[tp] 
    \centering
    \subfloat[]{\includegraphics[width=0.32\textwidth]{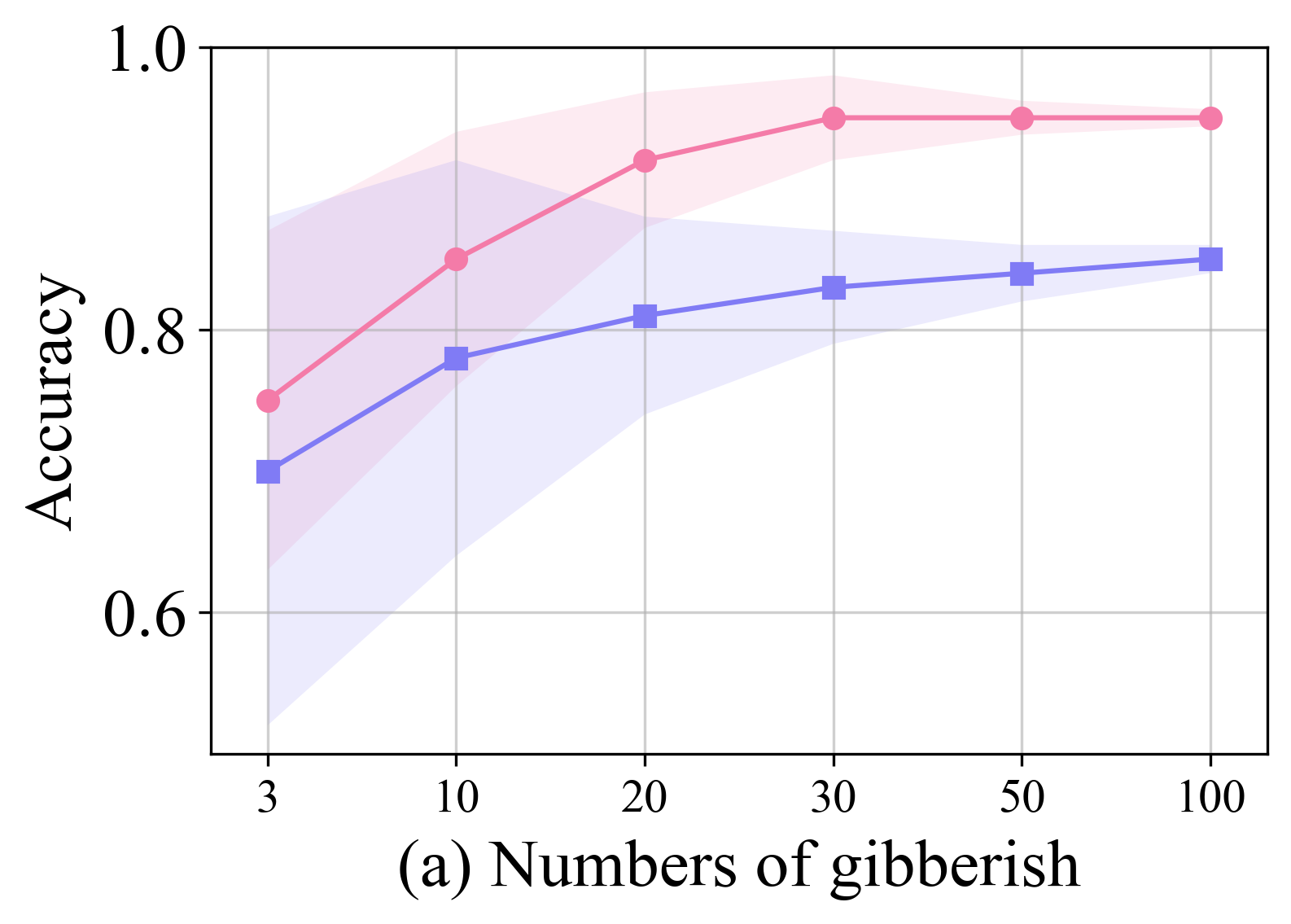}}
    \hfill
    \subfloat[]{\includegraphics[width=0.32\textwidth]{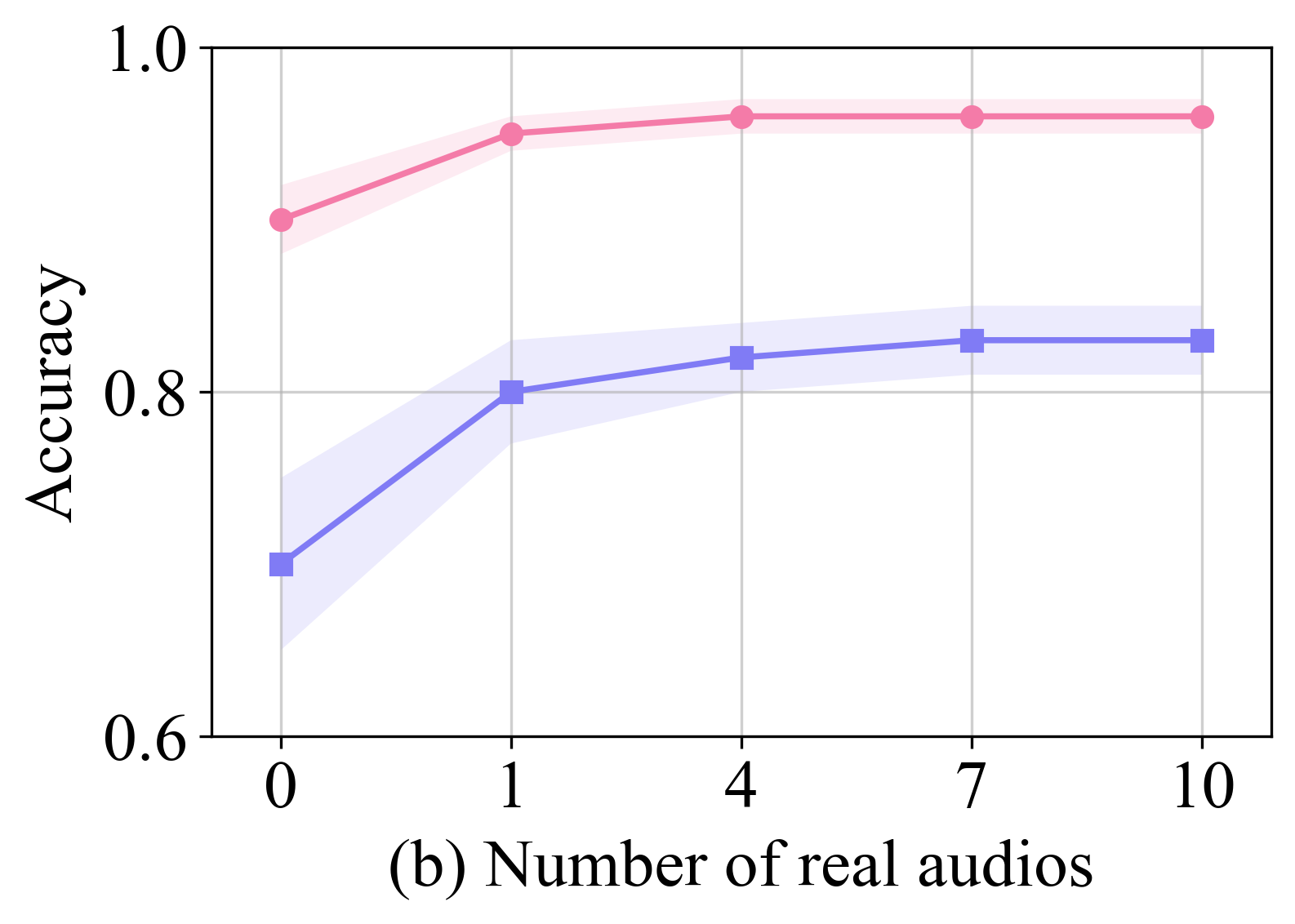}}
    \hfill
    \subfloat[]{\includegraphics[width=0.32\textwidth]{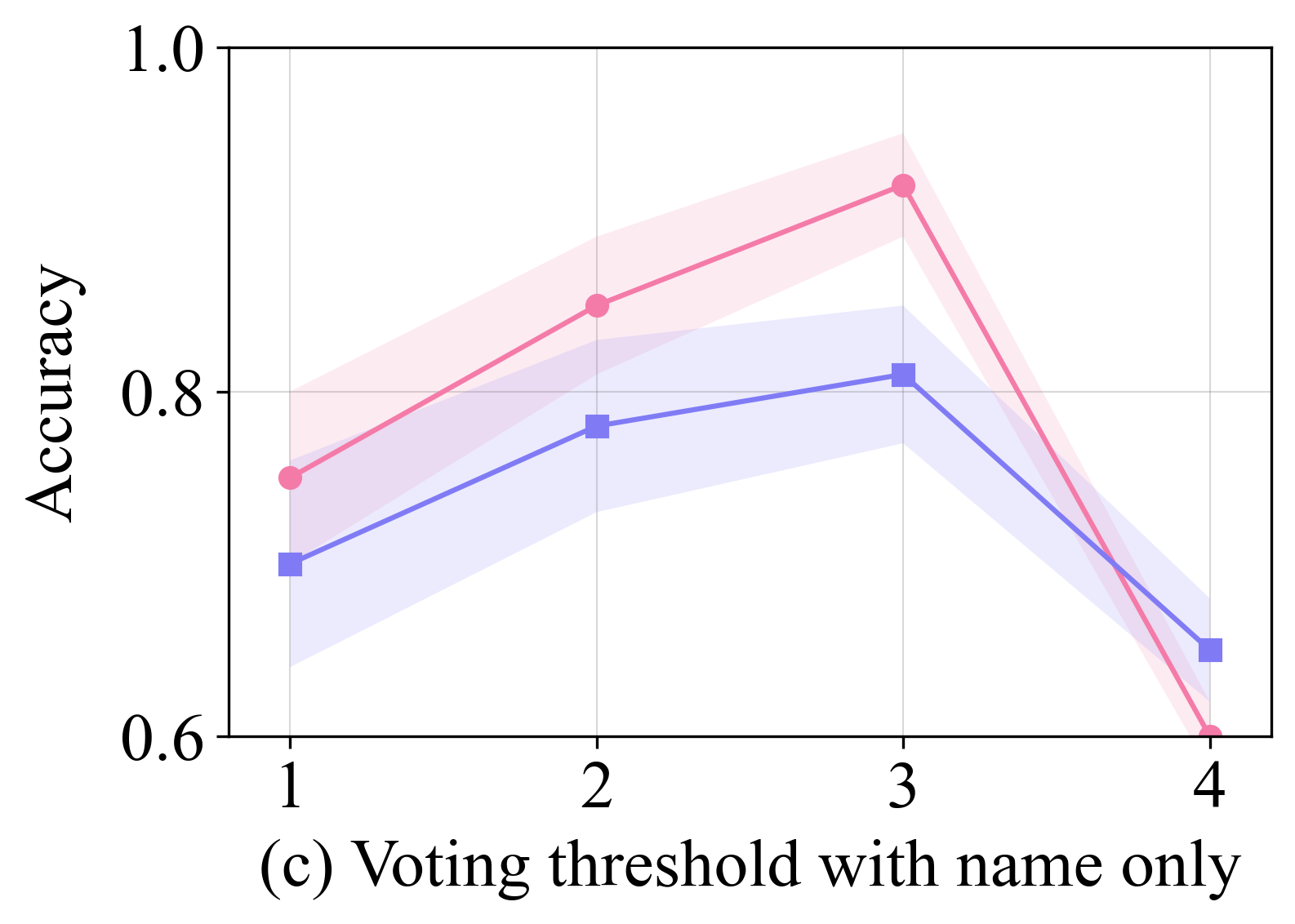}}
    \\[2mm]
    \subfloat[]{\includegraphics[width=0.32\textwidth]{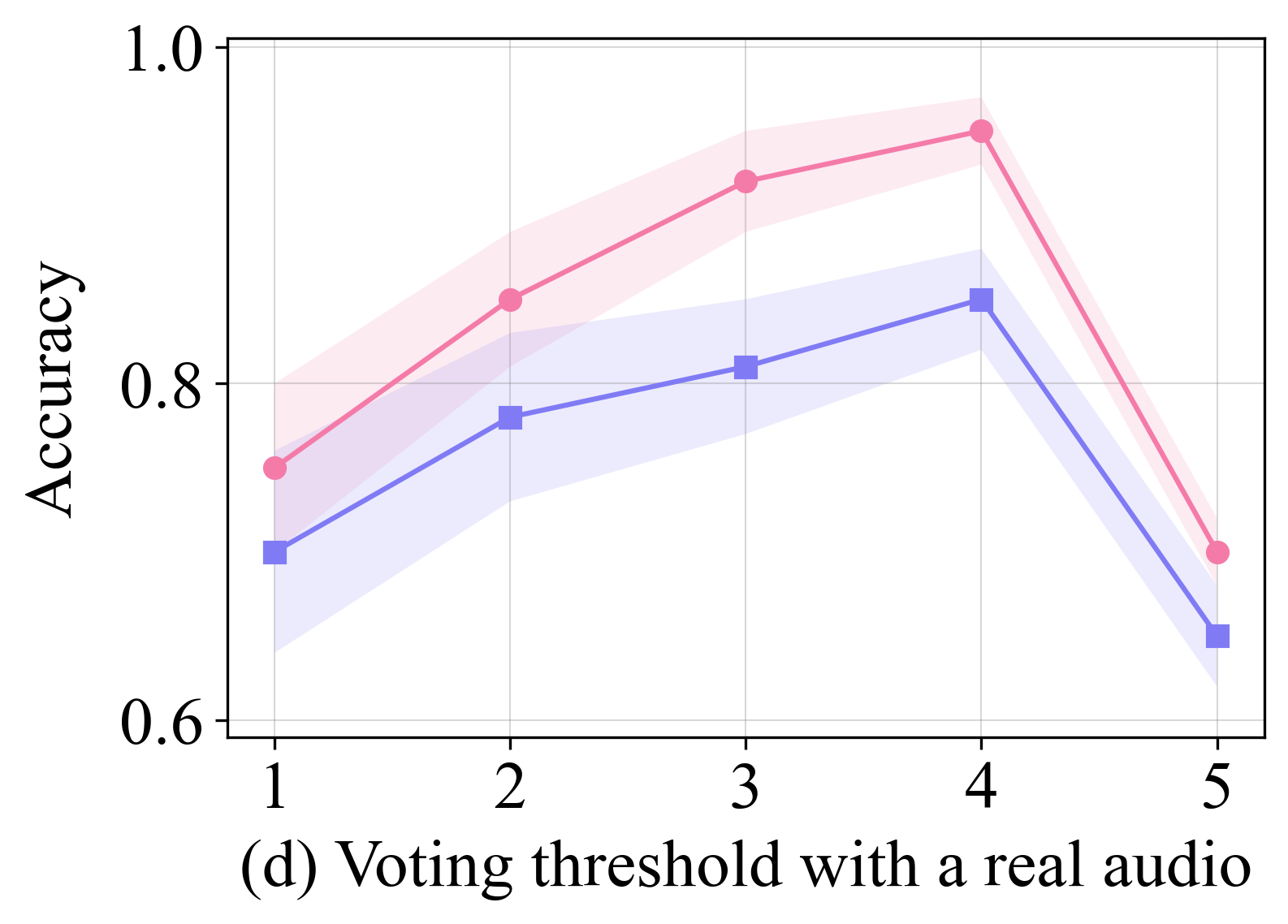}}
    \hfill
    \subfloat[]{\includegraphics[width=0.32\textwidth]{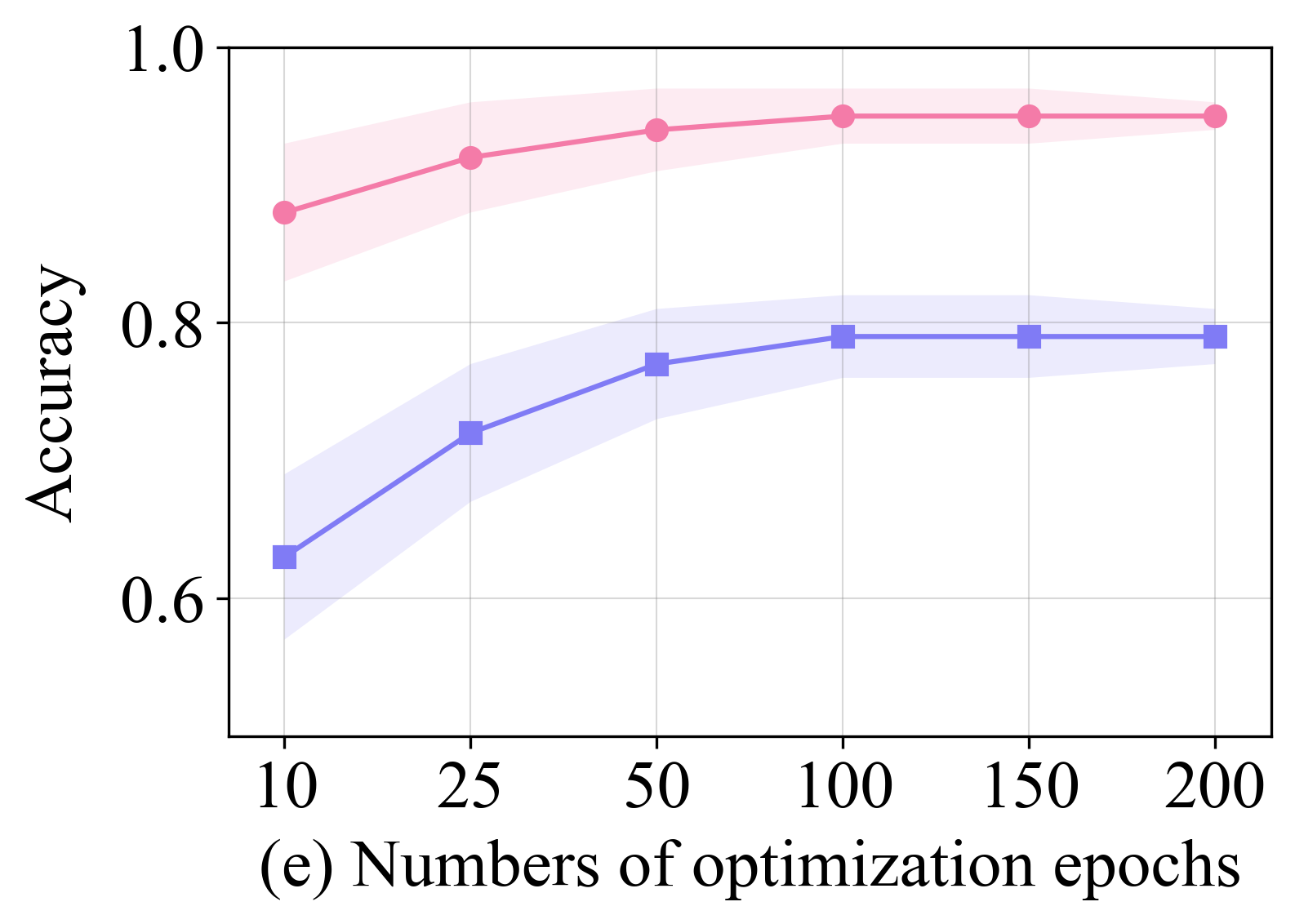}}
    \hfill
    \subfloat[]{\includegraphics[width=0.32\textwidth]{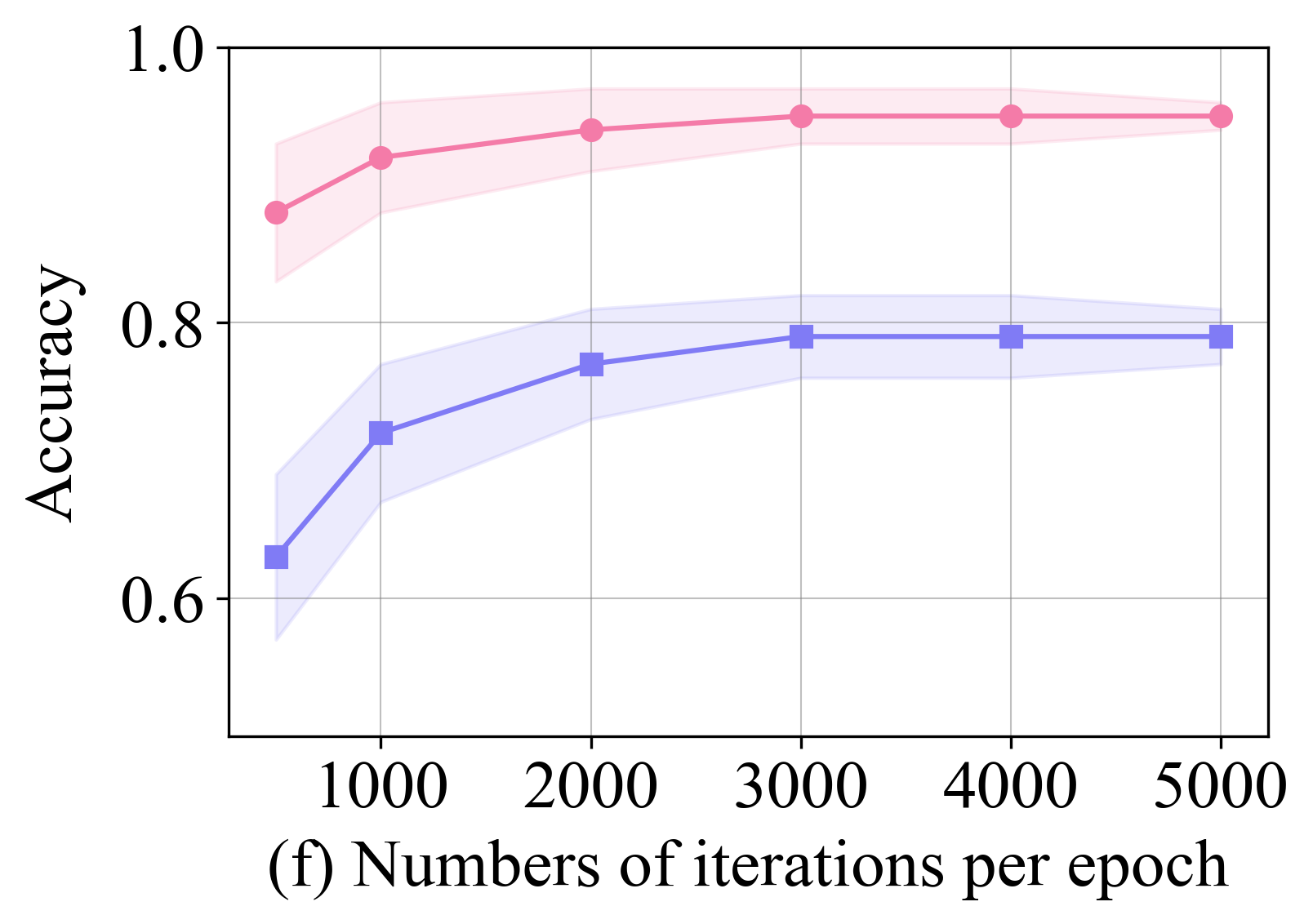}}
    
    \vspace{2mm}
    \includegraphics[width=0.5\textwidth]{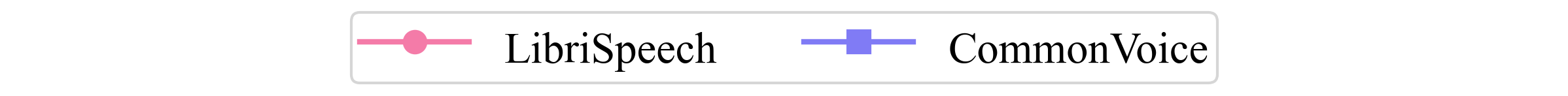}
    \caption{Detection accuracy for CLAP model (LibriSpeech) under various parameters.}
    \label{fig:clapablation}
\end{figure*}

\begin{figure*}[tp]
    \centering
    \includegraphics[width=0.9 \linewidth]{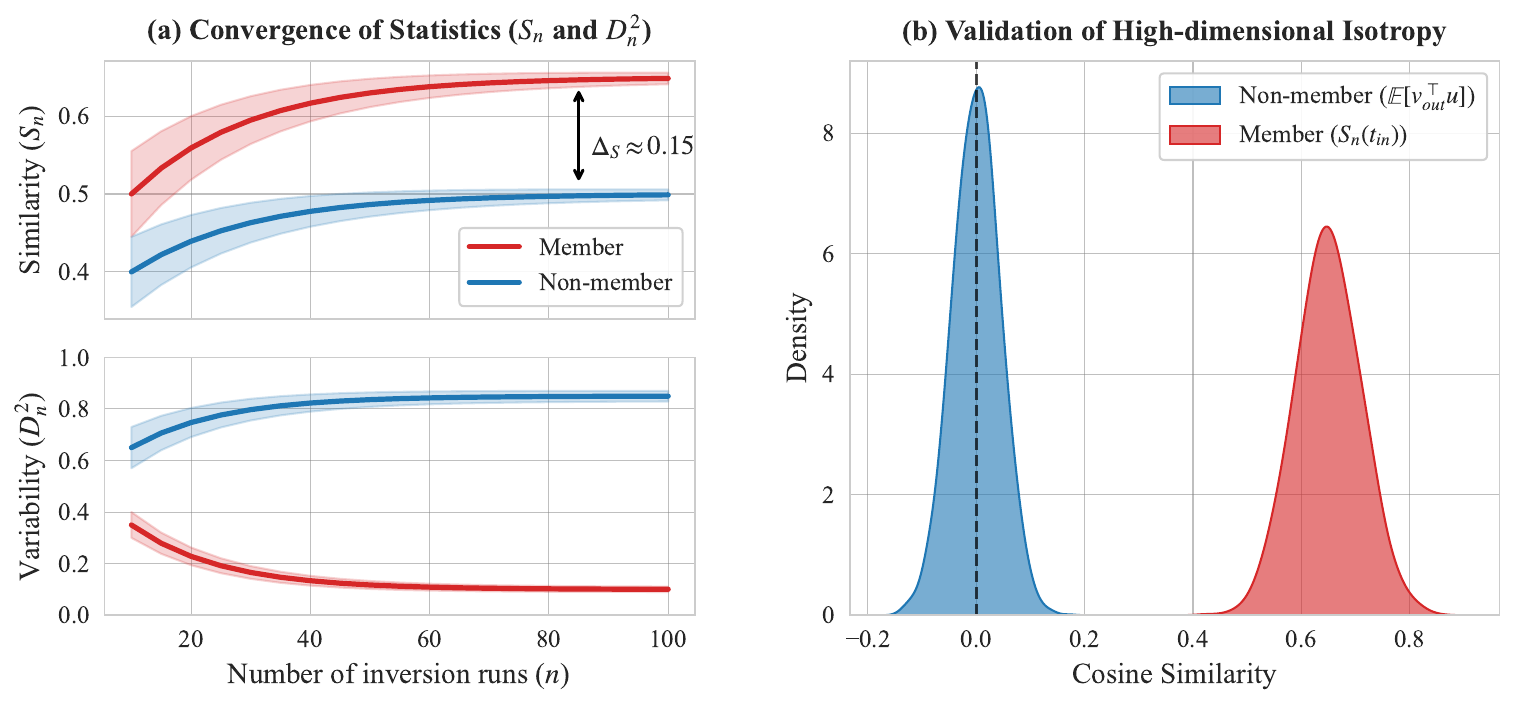} 
    \caption{Empirical validation of geometric separation. (a) Convergence of similarity ($S_n$) and variability ($D_n^2$) as inversion runs ($n$) increase. (b) Cosine similarity distributions confirming the high-dimensional isotropy of non-members versus the directional bias of members.}
    \label{fig:theory_validation}
\end{figure*}

\subsection{Ablation Study}

We conduct ablation studies in the following four perspectives to quantify how key design choices affect detection performance on both CLIP and CLAP (Figure~\ref{fig:clipablation} and Figure~\ref{fig:clapablation}).

\paragraph{Number of gibberish texts} Figure \ref{fig:clipablation}(a) and Figure \ref{fig:clapablation}(a) show that accuracy improves as the number of randomly generated gibberish strings ($\ell$) increases, and stabilizes once $\ell \ge 50$ for both modalities. With fewer than 20 strings, training becomes noisy and less reliable. Beyond 50 strings, the gain is marginal (typically $<0.5\%$ in accuracy), so we use $\ell=100$ as a practical choice that balances performance and overhead.

\paragraph{Real sample enhancement} As shown in Figure \ref{fig:clipablation}(b) and Figure \ref{fig:clapablation}(b), adding even a single real sample consistently improves accuracy (by roughly $1.5-3.5\%$, depending on the dataset and training set size). In contrast, using more than one real sample provides limited additional benefit (typically $<1\%$), suggesting that real samples are sufficient but not necessary to capture the modality-specific structure needed by the clustering-based vote.

\paragraph{Detection threshold analysis} We study the sensitivity of the ensemble decision to the voting threshold. As shown in Figure \ref{fig:clipablation}(c) and Figure \ref{fig:clapablation}(c), text-only inputs achieve the best accuracy with a three-vote threshold, whereas incorporating real samples favors a four-vote threshold (Figure \ref{fig:clipablation}(d) and Figure \ref{fig:clapablation}(d)). A higher threshold tends to miss true anomalies (lower recall), while a lower threshold increases false positives. The selected thresholds provide a robust balance between precision and recall across both modalities.

\paragraph{Optimization parameters} We examine two hyperparameters used during feature extraction: the number of epochs (Figure \ref{fig:clipablation}(e) and Figure \ref{fig:clapablation}(e) and the number of iterations per epoch (Figure \ref{fig:clipablation}(f) and Figure \ref{fig:clapablation}(f)). Overall, performance peaks at 100 epochs with 1000 iterations per epoch for CLIP, and 50 epochs with 1000 iterations per epoch for CLAP. Increasing iterations beyond 100 per epoch yields negligible gains (typically $<0.3\%$ in accuracy) while substantially increasing computation. Likewise, using more than 100 epochs leads to diminishing returns.

\subsection{Empirical Validation of Geometric Separation}
To corroborate our theoretical guarantees (Theorem 3.1), we empirically validate the geometric behaviors of members and non-members in the latent space (Figure \ref{fig:theory_validation}).

\textbf{Convergence of Statistics ($S_n$ and $D_n^2$):} 
Theoretically, members exhibit higher similarity ($S_n$) and lower variability ($D_n^2$) by concentrating on the memorized prototype. Empirically, as inversion runs ($n$) increase (Figure \ref{fig:theory_validation}a), the variance of $S_n$ diminishes, and the member/non-member gap ($\Delta_S$) strictly stabilizes at $\approx 0.15$. Conversely, non-members' $D_n^2$ converges to a significantly higher plateau, validating our dispersion bounds (Lemma B.5).

\textbf{Verification of Non-member Isotropy:} Assumption B.2 posits that non-members exhibit high-dimensional isotropy. We verify this by computing the cosine similarity between non-member embeddings ($v_{out}$) and random Gaussian vectors ($u$). As Figure \ref{fig:theory_validation}b shows, this distribution strictly centers at 0 with minimal variance, confirming the isotropy condition. In contrast, members display a strong directional bias, demonstrating that $S_n$ and $D_n^2$ effectively capture fundamental geometric disparities rather than random noise.

%% file: 5_defense_revised.tex
\section{Robustness against Potential Defenses}

To comprehensively assess the practical effectiveness of our approach, we evaluate UMID's robustness against potential defenses that model owners might deploy to mitigate membership inference attacks.

\begin{table}[t]
\renewcommand{\arraystretch}{0.9} 
\centering
\caption{UMID robustness evaluation under differential privacy defense ($\epsilon=1.0$, $\delta=10^{-5}$).}
\label{table:defense1}
\resizebox{0.5\textwidth}{!}{
\begin{tabular}{@{}ccc|c|c@{}}
\toprule
\multirow{2}{*}[-3pt]{\textbf{Dataset}} & \multirow{2}{*}[-3pt]{\textbf{Training samples}} & \multirow{2}{*}[-3pt]{\textbf{Defense}} & \textbf{Effectiveness} & \textbf{Efficiency} \\
\cmidrule(lr){4-4} \cmidrule(lr){5-5}
& & & \textbf{Accuracy} & \textbf{Time} \\
\midrule
\multirow{4}{*}[-3pt]{\shortstack{\textbf{CelebA}\\\textbf{(CLIP)}}} & \multirow{2}{*}{one-shot} & Without & .9535 & 0.685s \\
& & DP & .8742 & 0.688s \\
\cmidrule{2-5}
& \multirow{2}{*}{many-shot} & Without & .9562 & 0.691s \\
& & DP & .8595 & 0.689s \\
\midrule
\multirow{4}{*}[-3pt]{\shortstack{\textbf{LibriSpeech}\\\textbf{(CLAP)}}} & \multirow{2}{*}{one-shot} & Without & .9354 & 0.647s \\
& & DP & .8788 & 0.651s \\
\cmidrule{2-5}
& \multirow{2}{*}{many-shot} & Without & .9524 & 0.659s \\
& & DP & .8340 & 0.663s \\
\midrule
\multirow{4}{*}[-3pt]{\shortstack{\textbf{CommonVoice}\\\textbf{(CLAP)}}} & \multirow{2}{*}{one-shot} & Without & .8356 & 0.753s \\
& & DP & .7590 & 0.749s \\
\cmidrule{2-5}
& \multirow{2}{*}{many-shot} & Without & .8569 & 0.776s \\
& & DP & .7915 & 0.772s \\
\bottomrule
\end{tabular}
}
\end{table}

\begin{table}[t]
\renewcommand{\arraystretch}{0.9} 
\centering
\caption{UMID robustness evaluation under covert gibberish defense.}
\label{table:defense2}
\resizebox{0.5\textwidth}{!}{
\begin{tabular}{@{}ccc|c|c@{}}
\toprule
\multirow{2}{*}[-3pt]{\textbf{Dataset}} & \multirow{2}{*}[-3pt]{\textbf{Training samples}} & \multirow{2}{*}[-3pt]{\textbf{Defense}} & \textbf{Effectiveness} & \textbf{Efficiency} \\
\cmidrule(lr){4-4} \cmidrule(lr){5-5}
& & & \textbf{Accuracy} & \textbf{Time} \\
\midrule
\multirow{4}{*}[-3pt]{\shortstack{\textbf{CelebA}\\\textbf{(CLIP)}}} & \multirow{2}{*}{one-shot} & Without & .9535 & 0.685s \\
& & CG & .9412 & 0.689s \\
\cmidrule{2-5}
& \multirow{2}{*}{many-shot} & Without & .9562 & 0.691s \\
& & CG & .9385 & 0.687s \\
\midrule
\multirow{4}{*}[-3pt]{\shortstack{\textbf{LibriSpeech}\\\textbf{(CLAP)}}} & \multirow{2}{*}{one-shot} & Without & .9354 & 0.647s \\
& & CG & .9210 & 0.652s \\
\cmidrule{2-5}
& \multirow{2}{*}{many-shot} & Without & .9524 & 0.659s \\
& & CG & .9395 & 0.655s \\
\midrule
\multirow{4}{*}[-3pt]{\shortstack{\textbf{CommonVoice}\\\textbf{(CLAP)}}} & \multirow{2}{*}{one-shot} & Without & .8356 & 0.753s \\
& & CG & .8220 & 0.756s \\
\cmidrule{2-5}
& \multirow{2}{*}{many-shot} & Without & .8569 & 0.776s \\
& & CG & .8415 & 0.781s \\
\bottomrule
\end{tabular}
}
\end{table}


\paragraph{Gaussian noise} We apply differential privacy to embedding computations by adding calibrated Gaussian noise. Specifically, for each query, noise sampled from $\mathcal{N}(0, \sigma^2 I_d)$ is added to returned embeddings, where $\sigma$ is calibrated to achieve $(\epsilon, \delta)$-differential privacy with $\epsilon=1.0$ and $\delta=10^{-5}$. Table~\ref{table:defense1} reports the robustness of UMID against a Differential Privacy (DP) defense in both the image and audio modalities. After injecting calibrated Gaussian noise ($\epsilon=1.0$), UMID exhibits only a moderate degradation: accuracy drops by roughly $6\%-12\%$, depending on the dataset and the number of samples. Concretely, for CLIP (ResNet-50x4) on CelebA, UMID still achieves $87.42\%$ accuracy with 1 sample and $85.95\%$ with 75 samples under the defense. Likewise, for CLAP on LibriSpeech, performance remains high at $87.88\%$ and $83.40\%$, and on CommonVoice it remains competitive at $75.90\%$ and $79.15\%$. Importantly, runtime is essentially unchanged across all settings (averaging $0.7$s), indicating that DP introduces negligible additional computational overhead for the attacker. Overall, these results suggest that while DP noise reduces the fidelity of embedding-based leakage, UMID’s reliance on geometric vulnerabilities remains effective under a moderate privacy budget; in practice, stronger privacy budgets ($\epsilon < 1.0$) and/or complementary defenses may be necessary to fully secure these models.

\paragraph{Covert gibberish generation} In practice, a target model may incorporate input filters that flag overtly anomalous queries (e.g., nonsensical gibberish) and return misleading outputs, causing UMID to misclassify the presence of PII. To obtain more covert queries, we generate gibberish-like strings that remain visually and phonotactically plausible by replacing a small subset of characters with syllables drawn from a different language. For example, the detector can synthesize query texts by mixing English-style names with syllables from Arabic medical terminology. Concretely, we first prompt an LLM (e.g., GPT-3.5-turbo) to produce lists of common English initial and final syllables, then post-process the lists to remove duplicates and encourage coverage. We randomly concatenate syllables to form pseudo-English names (e.g., Karinix'', Zylogene'', ``Renotyl''), and explicitly verify their novelty against a name lexicon to avoid collisions with real entities. Finally, we prompt the LLM to instantiate these syllable templates into full strings, yielding covert gibberish that closely resembles authentic names (see Table~\ref{table:covertgibberish}). Table~\ref{table:defense2} summarizes UMID’s robustness to the Covert Gibberish (CG) defense for both CLIP and CLAP. UMID remains highly resilient: CG reduces accuracy by only $1.23\%$--$1.54\%$ across all datasets and training sample sizes. For example, on CelebA with 75 training samples, UMID still achieves $93.85\%$ accuracy under CG. Running time is also essentially unchanged (about $0.65$--$0.78$s), suggesting that CG neither introduces meaningful computational overhead nor effectively disrupts the identification procedure. Overall, UMID preserves both effectiveness and efficiency under covert, text-like perturbations.

%% file: 6_conclusion.tex
\section{Conclusion}
\label{sec:conclusion}
This work revisits membership inference for contrastive pretraining models through the lens of identity-level auditing under a unimodal privacy constraint, motivated by the practical limitations of existing MIAs: prohibitive shadow-model training and the need to query targets with paired biometric inputs. In this paper, we introduced UMID, a text-only auditing framework that avoids both requirements by leveraging text-guided cross-modal latent inversion. The key insight is that member identities induce reconstructions that are not only better aligned with the queried text but also more consistent across randomized inversions. UMID operationalizes this intuition via two statistics (similarity and variability) and a lightweight non-member reference built from synthetic gibberish, enabling membership decisions through an ensemble of unsupervised anomaly detectors. Extensive experiments on CLIP and CLAP across architectures and exposure regimes validate that UMID consistently outperforms prior baselines while remaining highly efficient, making it suitable for third-party auditing where sensitive biometrics are unavailable or impermissible to submit. Overall, UMID provides a practical and privacy-compliant pathway to assess PII memorization in widely deployed multimodal foundation models.


%% file: 7_contributions.tex
\section*{Author Contributions}
\label{7_contributions}

Ruoxi Cheng contributed to the conceptualization and methodology, performed the preliminary experimental validation, and was responsible for writing the original draft. Yizhong Ding conducted the large-scale experimental validation and formal analysis. Jian Zhao and Tianle Zhang were responsible for the methodology refinement and validation. Hongyi Zhang and Haoxuan Ma contributed to the writing, reviewing, and editing of the manuscript. Yiyan Huang provided expertise in mathematical modeling and formal analysis, contributed to the writing, reviewing, and editing, and was responsible for funding acquisition and supervision. Xuelong Li provided conceptual guidance, project administration, and overall supervision. All authors have read and agreed to the published version of the manuscript.

%% file: 9_appendix.tex
\clearpage
\appendix

\section{Gibberish Samples}
\label{sec:gibberish}

Table \ref{table:gibberish} and Table \ref{table:covertgibberish} present examples of randomly generated gibberish and covert gibberish that mimics authentic names, respectively.

\begin{table}[!h]
\centering
\footnotesize 
\caption{Samples of randomly generated textual gibberish generated by GPT-3.5-turbo.}
\label{table:gibberish}
\renewcommand{\arraystretch}{1.3} 
\setlength{\tabcolsep}{6pt} 
\begin{tabular}{@{}lll@{}}
\toprule
b+dh43u!f9de545w & 53e(s3erg24=pnI<S & fe3\_5fg;@f34f.e2w/ \\
\midrule
4teh<E\{434fef234ter & 5gt\%b@5435-hgF & \#4c3rd4scd324g\\
\bottomrule
\end{tabular}
\end{table}

\begin{table}[!h]
\centering
\footnotesize 
\caption{Covert gibberish that seem to be real names generated by GPT-3.5-turbo.}
\label{table:covertgibberish}
\renewcommand{\arraystretch}{1.2} 
\setlength{\tabcolsep}{8pt} 
\begin{tabular}{@{}lllll@{}}
\toprule
Karinix  & Zylogene & Glycogenyx & Zylotrax & Vexilith\\
\midrule 
Exodynix & Novylith & Glycosyne  & Xenolynx & Rynexis\\
\bottomrule
\end{tabular}
\end{table}

\section{Theoretical Guarantees for Geometric Separation}
\label{sec:appendix_proof}

In this appendix, we provide the detailed proofs for the theoretical claims made in Section \ref{sec:theory}. We establish that under mild geometric assumptions about the contrastive latent space, the proposed statistics $S_n$ and $D_n^2$ fundamentally separate member identities from non-members with high probability.

\subsection{Preliminaries and Assumptions}

We fix a text $t$ with embedding $v_t = \phi(t) \in \mathbb{R}^d$. The randomized modality inversion produces $n$ independent optimized embeddings $v^{(1)}, \dots, v^{(n)}$. Recall the definitions of our statistics:
\begin{align}
    \bar{v}_n &:= \frac{1}{n}\sum_{i=1}^n v^{(i)}, \quad S_n := v_t^\top \bar{v}_n, \nonumber \\
    D_n^2 &:= \frac{1}{n} \sum_{i=1}^n \|v^{(i)}-\bar{v}_n\|_2^2.
\end{align}

We model the latent space geometry using $M$ unit-norm prototypes $\{\mu_k\}_{k=1}^M \subset \mathbb{S}^{d-1}$. We restate the core assumptions formally for the proof:

\begin{assumption}[Prototype orthogonality] \label{assump:proto}
	The prototypes $\mu_1,\dots,\mu_M$ are approximately orthogonal. Specifically, let $\rho_d := \max_{y\neq z} |\mu_y^\top \mu_z|$. For $\delta_1 \in (0,1)$, with probability at least $1-\delta_1$, $\rho_d \le \sqrt{\frac{4\log(2M^2/\delta_1)}{d}}$.
\end{assumption}

\begin{assumption}[Text-prototype geometry] \label{assump:geometry_app}
\textbf{Non-member:} For a non-member $t_{\mathrm{out}}$, its embedding $v_{\mathrm{out}}$ is isotropic relative to random directions. For any unit $u$, $\mathbb{P}(|v_{\mathrm{out}}^\top u| \ge \epsilon) \le 2\exp(-c d \epsilon^2)$. \textbf{Member:} For a member $t_{\mathrm{in}}$, there exists a prototype $y^\star$ such that $v_{\mathrm{in}}^\top \mu_{y^\star} \ge \gamma_{\mathrm{in}} > 0$.
\end{assumption}

\begin{assumption}[Randomized Optimization] \label{assump:opt_app}
	Each run $i$ selects a prototype index $C^{(i)} \sim p(t)$ and lands near it: $v^{(i)} = \mu_{C^{(i)}} + \Delta^{(i)}$, with $\|\Delta^{(i)}\|_2 \le \varepsilon_{\mathrm{opt}}$ with high probability. \textbf{For members}, $p(t_{\mathrm{in}})$ concentrates on $y^\star$: $p_{y^\star} \ge 1 - \delta^\star$ (where $\delta^\star$ is small). \textbf{For non-members}, $p(t_{\mathrm{out}})$ is nearly uniform: $p_y \approx 1/M$.
\end{assumption}

\subsection{Population-Level Separation}

We first analyze the population statistics defined as $n \to \infty$ (assuming $\varepsilon_{\mathrm{opt}} \to 0$):
$S_\infty(t) := v_t^\top m(t)$ and $D_\infty^2(t) := 1 - \|m(t)\|_2^2$, where $m(t) := \sum_{y=1}^M p_y(t)\mu_y$ is the mean prototype vector.

\begin{proposition}[Population separation gaps] \label{prop:pop_gaps}
	Under the assumptions above, with high probability over prototypes:
	\begin{enumerate}
		\item \textbf{Member:} $S_\infty(t_{\mathrm{in}}) \ge \gamma_{\mathrm{in}} - 2\delta^\star$ and $D_\infty^2(t_{\mathrm{in}}) \le 2\delta^\star + 3\rho_d \delta^\star \approx 0$.
		\item \textbf{Non-member:} $|S_\infty(t_{\mathrm{out}})| \le O(d^{-1/2}) \approx 0$ and $D_\infty^2(t_{\mathrm{out}}) \ge 1 - \frac{1}{M} - \rho_d \approx 1$.
	\end{enumerate}
\end{proposition}

\begin{proof}
	\textbf{For Member:} 
	$S_\infty = \sum p_y v_{\mathrm{in}}^\top \mu_y = p_{y^\star} v_{\mathrm{in}}^\top \mu_{y^\star} + \sum_{y \neq y^\star} p_y v_{\mathrm{in}}^\top \mu_y$. 
	Using $p_{y^\star} \ge 1-\delta^\star$, $v_{\mathrm{in}}^\top \mu_{y^\star} \ge \gamma_{\mathrm{in}}$, and trivial bounds $|v^\top \mu| \le 1$, we get $S_\infty \ge (1-\delta^\star)\gamma_{\mathrm{in}} - \delta^\star \approx \gamma_{\mathrm{in}}$.
	For dispersion, $D_\infty^2 = 1 - \|p_{y^\star}\mu_{y^\star} + \sum_{y\neq y^\star} p_y \mu_y\|_2^2$. The cross-terms are bounded by $\rho_d$. Dominant term is $1 - p_{y^\star}^2 \approx 1-(1-\delta^\star)^2 \approx 2\delta^\star$.
	
	\textbf{For Non-member:}
	$S_\infty = v_{\mathrm{out}}^\top m(t_{\mathrm{out}})$. Since $v_{\mathrm{out}}$ is isotropic and independent of $m(t_{\mathrm{out}})$, $S_\infty$ concentrates around 0 with rate $d^{-1/2}$ (Assumption \ref{assump:geometry_app}).
	For dispersion, $\|m(t_{\mathrm{out}})\|_2^2 = \|\sum p_y \mu_y\|_2^2 = \sum p_y^2 + \sum_{y \neq z} p_y p_z \mu_y^\top \mu_z$. With $p_y \approx 1/M$, $\sum p_y^2 \approx 1/M$. Cross terms are bounded by $\rho_d (\sum p_y)^2 = \rho_d$. Thus $D_\infty^2 \ge 1 - (1/M + \rho_d)$.
\end{proof}

\subsection{Finite-Sample Concentration}

We now show that empirical statistics $(S_n, D_n^2)$ converge to population values as $O(1/\sqrt{n})$.
Decompose the empirical mean $\bar v_n = \tilde m_n + \bar \Delta_n$, where $\tilde m_n = \frac{1}{n}\sum \mu_{C^{(i)}}$ is the average of selected prototypes and $\bar \Delta_n$ is the average optimization residual.

\begin{lemma}[Concentration Bounds] \label{lem:conc}
	Conditioned on small residuals $\|\Delta^{(i)}\| \le \varepsilon_{\mathrm{opt}}$, for any $\epsilon > 0$:
	\begin{align}
		\mathbb{P}(|S_n - S_\infty| \ge \epsilon + \varepsilon_{\mathrm{opt}}) &\le 2\exp(-2n\epsilon^2), \\
		\mathbb{P}(|D_n^2 - D_\infty^2| \ge 4\epsilon + C\varepsilon_{\mathrm{opt}}) &\le 2\exp(-n\epsilon^2).
	\end{align}
\end{lemma}

\begin{proof}
	\textbf{For $S_n$:} 
	$S_n = v_t^\top \tilde m_n + v_t^\top \bar \Delta_n$. The term $|v_t^\top \bar \Delta_n| \le \varepsilon_{\mathrm{opt}}$. The term $v_t^\top \tilde m_n = \frac{1}{n}\sum Z_i$ where $Z_i = v_t^\top \mu_{C^{(i)}}$ are i.i.d. variables bounded in $[-1, 1]$ with mean $S_\infty$. Hoeffding's inequality gives the result.
	
	\textbf{For $D_n^2$:}
	Using the identity $D_n^2 = \frac{1}{n}\sum \|v^{(i)}\|^2 - \|\bar v_n\|^2$, and decomposing $v^{(i)}$, we can bound $|D_n^2 - D_\infty^2| \le 2\|\tilde m_n - m(t)\|_2 + O(\varepsilon_{\mathrm{opt}})$. Since $\tilde m_n$ is the average of bounded i.i.d. vectors with mean $m(t)$, Hoeffding inequality yields the exponential tail bound.
\end{proof}

\subsection{Proof of Theorem \ref{thm:main}}

The goal is to show that with probability at least $1-\delta$, the empirical statistics $(S_n, D_n^2)$ for members and non-members fall on opposite sides of the thresholds $s_{\mathrm{thr}}$ and $d_{\mathrm{thr}}^2$. We structure the proof in three logical steps: (1) ensuring valid population geometry, (2) bounding the optimization and sampling errors, and (3) deriving the separation condition.

We allocate the total failure probability $\delta$ into three components: prototype geometry failure ($\delta_{\mathrm{proto}}$), optimization localization failure ($\delta_{\mathrm{opt}}'$), and sampling concentration failure ($\delta_{\mathrm{samp}}$), such that $\delta_{\mathrm{total}} = \delta_{\mathrm{proto}} + \delta_{\mathrm{opt}}' + \delta_{\mathrm{samp}} \le \delta$. We set each component to $\delta/4$.

\textbf{(1) Existence of population gaps.}
By Proposition \ref{prop:pop_gaps} and the prototype orthogonality assumption (Assumption \ref{assump:proto}), there exists an event $\mathcal{E}_{\mathrm{proto}}$ with $\mathbb{P}(\mathcal{E}_{\mathrm{proto}}) \ge 1-\delta/4$ where the population gaps are strictly positive. We define the minimum margin $\Gamma := \min\left\{ S_\infty(t_{\mathrm{in}}) - S_\infty(t_{\mathrm{out}}), \quad D_\infty^2(t_{\mathrm{out}}) - D_\infty^2(t_{\mathrm{in}}) \right\} > 0$.
The decision thresholds are defined as the midpoints: $s_{\mathrm{thr}} = \frac{1}{2}(S_\infty(t_{\mathrm{in}}) + S_\infty(t_{\mathrm{out}}))$ and $d_{\mathrm{thr}}^2 = \frac{1}{2}(D_\infty^2(t_{\mathrm{in}}) + D_\infty^2(t_{\mathrm{out}}))$.

\textbf{(2) Concentration of empirical statistics.}
We require the empirical statistics to concentrate around their population means within a radius smaller than $\Gamma/2$.
First, consider the \textit{optimization localization}. Let $\mathcal{E}_{\mathrm{opt}}(t)$ be the event that all $n$ runs satisfy $\|\Delta^{(i)}\|_2 \le \varepsilon_{\mathrm{opt}}$. By the union bound over $2n$ runs (for both $t_{\mathrm{in}}$ and $t_{\mathrm{out}}$), we have $\mathbb{P}(\mathcal{E}_{\mathrm{opt}}^c) \le 2n\delta_{\mathrm{opt}}$. We require $2n\delta_{\mathrm{opt}} \le \delta/4$.

Next, consider the \textit{sampling noise}. Conditioned on $\mathcal{E}_{\mathrm{opt}}$, Lemma \ref{lem:conc} guarantees concentration. We choose an error tolerance $\epsilon$ and residual $\varepsilon_{\mathrm{opt}}$ such that the total deviation is bounded by $\Gamma/4$. Specifically, let the target deviation bound be $\Delta_{\mathrm{dev}} := 4\epsilon + C_\Delta \varepsilon_{\mathrm{opt}}$. We require $\Delta_{\mathrm{dev}} \le \Gamma/2$.
By Lemma \ref{lem:conc}, the probability that any of the four statistics ($S_n, D_n^2$ for member/non-member) deviates by more than $\Delta_{\mathrm{dev}}$ is bounded by $8\exp(-n\epsilon^2)$. Setting this to $\delta/4$ yields the sample complexity requirement:
\[
n \ge \frac{1}{\epsilon^2} \log\left(\frac{32}{\delta}\right) = \Omega\left(\Gamma^{-2} \log(1/\delta)\right).
\]

\textbf{(3) Geometric separation condition.}
Let $\mathcal{E}_{\mathrm{good}}$ be the intersection of valid geometry, optimization localization, and sampling concentration. On this event, for the similarity statistic $S_n(t_{\mathrm{in}})$:
\begin{align*}
	S_n(t_{\mathrm{in}}) - s_{\mathrm{thr}} 
	&= S_n(t_{\mathrm{in}}) - \frac{1}{2}\left(S_\infty(t_{\mathrm{in}}) + S_\infty(t_{\mathrm{out}})\right) \\
	&= \frac{1}{2}\underbrace{\left(S_\infty(t_{\mathrm{in}}) - S_\infty(t_{\mathrm{out}})\right)}_{\ge \Gamma} - \underbrace{\left|S_n(t_{\mathrm{in}}) - S_\infty(t_{\mathrm{in}})\right|}_{\le \Delta_{\mathrm{dev}}} \\
	&\ge \frac{\Gamma}{2} - \frac{\Gamma}{2} = 0 \implies S_n(t_{\mathrm{in}}) \ge s_{\mathrm{thr}}.
\end{align*}
Similarly, for the non-member, $s_{\mathrm{thr}} - S_n(t_{\mathrm{out}}) \ge \frac{\Gamma}{2} - \Delta_{\mathrm{dev}} \ge 0$, implying $S_n(t_{\mathrm{out}}) \le s_{\mathrm{thr}}$.
The exact same logic applies to the dispersion $D_n^2$:
\begin{align*}
	d_{\mathrm{thr}}^2 - D_n^2(t_{\mathrm{in}}) \ge \frac{\Gamma}{2} - \Delta_{\mathrm{dev}} \ge 0 \implies D_n^2(t_{\mathrm{in}}) \le d_{\mathrm{thr}}^2.
\end{align*}
Thus, on $\mathcal{E}_{\mathrm{good}}$, the separation conditions hold simultaneously. Together, the probability of failure is bounded by the sum of failure probabilities for each step: $\mathbb{P}(\mathcal{E}_{\mathrm{good}}^c) \le \delta_{\mathrm{proto}} + \delta_{\mathrm{opt}}' + \delta_{\mathrm{samp}} \le \delta$. This completes the proof. \qed


